\documentclass[journal]{IEEEtran}
\usepackage{amssymb}
\usepackage{graphicx,cite,epsfig,color, amssymb,amsmath,multirow,enumerate,amsthm}
\usepackage{paralist,  xcolor}
\usepackage[normalem]{ulem}

\begin{document}


\def\QEDclosed{\mbox{\rule[0pt]{1.3ex}{1.3ex}}}

\def\QEDopen{{\setlength{\fboxsep}{0pt}\setlength{\fboxrule}{0.2pt}\fbox{\rule[0pt]{0pt}{1.3ex}\rule[0pt]{1.3ex}{0pt}}}}
\def\QED{\QEDopen}

\def\proof{}
\def\endproof{\hspace*{\fill}~\QED\par\endtrivlist\unskip}

\title{Dealing with Interference in Distributed Large-scale MIMO Systems: A Statistical Approach}

\author{Haifan~Yin,
    David~Gesbert~\IEEEmembership{Fellow,~IEEE},
    Laura~Cottatellucci%
\thanks{This work was supported by the Seventh Framework Programme for Research of the European Commission under grant number HARP-318489.}%
\thanks{Copyright (c) 2013 IEEE. Personal use of this material is permitted. However, permission to use this material for any other purposes must be obtained from the IEEE by sending a request to pubs-permissions@ieee.org.}%
\thanks{H. Yin, D. Gesbert, and L. Cottatellucci are with EURECOM, 06410 Biot, France (e-mail: yin@eurecom.fr, gesbert@eurecom.fr, cottatel@eurecom.fr).}%
}





\IEEEpubid{0000--0000/00/\$31.00~\copyright~2013 IEEE}

\maketitle

\theoremstyle{plain}
\newtheorem{theorem}{Theorem}
\theoremstyle{plain}
\newtheorem{proposition}{Proposition}
\theoremstyle{plain}
\newtheorem{corollary}{Corollary}

\begin{abstract}
\boldmath
This paper considers the problem of interference control through the use of second-order statistics in massive MIMO multi-cell networks. We consider both the cases of co-located massive arrays and large-scale distributed antenna settings. We are interested in characterizing the low-rankness of users' channel covariance matrices, as such a property can be exploited towards improved channel estimation (so-called pilot decontamination) as well as interference rejection via spatial filtering. In previous work, it was shown that massive MIMO channel covariance matrices exhibit a useful finite-rank property that can be modeled via the angular spread of multipath at a MIMO uniform linear array. This paper extends this result to more general settings including certain non-uniform arrays, and more surprisingly, to two dimensional distributed large scale arrays. In particular our model exhibits the dependence of the signal subspace's richness on the scattering radius around the user terminal, through a closed form expression. The applications of the low-rankness covariance property to channel estimation's denoising and low-complexity interference filtering are highlighted.


\end{abstract}

\begin{IEEEkeywords}
massive MIMO, distributed antennas, channel estimation, interference mitigation, covariance matrix.
\end{IEEEkeywords}

\IEEEpeerreviewmaketitle

\section{Introduction}

\IEEEPARstart{F}{ull} spatial reuse of the frequency resource across even neighboring cells is a {\em de facto} standard approach in wireless network design. The downside of this strategy lies in the high amount of inter-cell interference, which in turn severely limits the performance of certain users, especially at cell-edge.
This fact has fueled extensive research on interference management, and particularly on methods relying on the use of spatial filtering at the base station side.
Recently, two schools of thought have emerged with conflicting strategies for how to best exploit the added spatial dimension offered by multiple-input multiple-output (MIMO) antennas.
In the first, the focus is on strengthening local beamforming capabilities by endowing each base station with a massive number of antenna elements that is substantially larger than the number of terminals served in the same cell on any given spectral resource block.
The added cost of hardware is compensated by the fact that simple distributed beamforming schemes that require little inter-cell cooperation can efficiently mitigate interference \cite{marzetta2010, rusek2013, hoydis2013, larsson2013}. In the second school, cooperation between cells is emphasized as the key towards increasing the spatial degrees of freedom \cite{gesbert2010JSAC}. In the cooperation approach, so-called network MIMO (or CoMP in the 3GPP terminology) schemes mimic the transmission over a virtual MIMO array encompassing the {\em spatially distributed} base station antennas. In contrast with the massive MIMO solution, the cooperative spatial filtering of interference is made possible with no {\em additional} antennas at the base station side, yet it goes at the expense of fast signaling links over the backhaul, a need for tight synchronization, and seemingly multi-user detection schemes that are computationally more demanding than the simple matched filters advocated in massive MIMO.
Additionally, a major hurdle preventing from realizing the full gains of MIMO multi-cell cooperation lies in the cost of acquiring and sharing channel estimates using orthogonal training sequences over large clusters \cite{Lozano2013}.

\IEEEpubidadjcol

Despite these differences, a fundamental common feature behind each philosophy lies in the {\em coherent combining} of a large number of antennas in view of interference nulling. Additionally, in both cases, our ability to reject interference is only as good as our ability to estimate the user channels properly. In the context of co-located massive MIMO, channel estimation from pilots that are inevitably reused over space leads to the so-called pilot contamination effect \cite{jose2011TWC, ngo2011ICASSP}.
Although initially branded as a fundamental limit of massive MIMO communications, a finer impact analysis of pilot contamination indicates that it is only one of several limitations of such systems \cite{bjornson2013}. When it comes to improving channel estimation, several possible solutions were recently proposed in a series of papers \cite{yin2012jsac, muller2013arXiv, ngo2013multicell}. In \cite{yin2012jsac}, an approach to de-interfere channel estimates was revealed based on the exploitation of second-order statistical properties of the received vector signal. The key enabler is the finite-rankness of the channels' covariance matrices which was shown to occur {\em in the asymptotic massive MIMO regime} whenever the angle spread of incoming/departing paths at the MIMO array is {\em limited}.
Independently, a similar finite-rank property was shown to be useful in the context of low-complexity scheduling and spatial beamforming for massive MIMO networks \cite{caire2012joint}.
Hence the low-dimensional property for the signal subspace (i.e. in which the MIMO channel realizations live) is instrumental to spatial interference rejection. These results were all reached for the case of uniform (equi-spaced) calibrated linear arrays. A natural question then arises as to whether the low-rank property can be established and exploited in more general large-scale antenna settings, such as random and two dimensional antenna placements.
This paper is devoted to this problem.

A first examination of \cite{yin2012jsac,caire2012joint} indicates that the finite-rank behavior is rooted in the asymptotic orthogonality between Fourier transform vectors corresponding to different path angles, suggesting the property might be restricted to the use of one-dimensional equi-spaced arrays. However our results point otherwise, showing low-rankness of channel's subspace for large-scale antenna systems is a recurrent trend applying to random and also distributed antenna placements, hence with a wider applicability to cooperative networks.

Our specific contributions are as follows:
First we consider a uniform linear massive array scenario yet with {\em several clusters} of multipath. In this case we establish a finite-rank model for the channel's covariance that directly extends that of \cite{yin2012jsac}, where the rank is shown to be a function of the incoming/departing angular spread of multipath. We then show that a similar low-rank result holds for a linear array with random placement of antenna elements. Although in this case, unlike the uniform array, the finite rank is only characterized by an upper bound. We show how this property can be used towards, for instance, pilot decontamination.


In the second part of the paper, we turn to a large-scale antenna regime where the antenna elements are scattered randomly throughout the (dense) network, yet can still be combined coherently. Such a setting with spatially distributed antennas includes remote radio head (RRH) networks, network-MIMO (CoMP) schemes with large clusters, and cloud-enabled radio access networks (C-RAN) as particular cases. A channel model building on the classical one-ring multipath model \cite{jakes1974, Shiu2000} is proposed to analyze this scenario.
In this setting we show that, there again surprisingly, the channel covariance exhibits a low-dimensional signal subspace behavior, in the large number of base station antenna regime, {\em even discounting path loss effects}.
We show the richness of the covariance's signal subspace is primarily governed by the scattering radius around the user terminal. We provide a closed form expression for an upper-bound of the covariance rank and show by simulation how this bound closely matches reality. Note that the notion that the total perimeter occupied by scatterers can govern the rank of the signal subspace in a distributed MIMO antenna setting is reminiscent of a previously observed phenomenon in the different context of compact MIMO arrays. In \cite{kennedy2007}, the authors establish a physical model for the dimension of the spatial multipath field of a disk-shaped compact area filled with MIMO antennas and illuminated by isotropic multipaths.


In the last part of the paper, we turn our attention to the exploitation of signal-subspace's low-rankness towards interference rejection for a distributed array.
We derive a lower bound on the signal to interference ratio that would be obtained in a two user setting with a simple matched filter, as a function of the distance between the users and the number of antennas. We show how a distance of two scattering radiuses can be selected as a critical minimal distance between selected co-channel users in a scheduling algorithm so as to facilitate interference nulling. As an application of the low-rankness property, a simple subspace-based interference mitigation scheme is put forward, which exploits the statistical information of the interference channels. Numerical results 
are presented in the last section.

The notations adopted in the paper are as follows. We use boldface to denote matrices and vectors. Specifically, ${\mathbf{I}}_M$ denotes the $M \times M$ identity matrix. ${({\mathbf{X}})^T}$, ${({\mathbf{X}})^*}$, and ${({\mathbf{X}})^H}$ denote the transpose, conjugate, and conjugate transpose of a matrix ${\mathbf{X}}$ respectively. $\mathbb{E}\left\{ \cdot \right\}$ denotes the expectation, ${\left\| \cdot \right\|_F}$ denotes the Frobenius norm.
The Kronecker product of two matrices ${\bf{X}}$ and ${\bf{Y}}$ is denoted by ${\bf{X}}\otimes{\bf{Y}}$. $\mbox{span} \{{\mathbf{v}_1, \mathbf{v}_2, ..., \mathbf{v}_n}\}$ is the span of linear vector space on the basis of $\mathbf{v}_1, \mathbf{v}_2, ..., \mathbf{v}_n$ for some $n\geq 1$, $\mbox{dim} \{\mathcal{A}\}$ is the dimension of a linear space $\mathcal{A}$, and $\mbox{null} \{\mathbf{R}\}$ is the null space of matrix $\mathbf{R}$. ${\mathop{\rm diag}\nolimits} \{ {\bf{a_1,...,a_N}}\}$ denotes a diagonal matrix or a block diagonal matrix with $\bf{a_1,...,a_N}$ at the main diagonal. $\triangleq$ is used for definition.

\section{Co-located massive linear arrays}\label{modeling}
We consider the uplink\footnote{Similar principles would apply in the downlink, which for ease of exposition is ignored here.} of a network of $B$ time-synchronized cells, with full spectrum reuse. Each of the $B$ base stations is equipped with a one-dimensional array of $M$ antennas, where $M$ is allowed to grow large (massive MIMO regime). For ease of exposition, all user terminals are assumed to be equipped with a single antenna. Furthermore we consider that a single user is served per cell and per resource block. A classical multipath model is given by \cite{molisch2010}:
\begin{equation}\label{Eq:simpleChanModel}
{{\mathbf{h}}_{i}} = \sqrt{ \frac{\beta_i}{ P }} \sum\limits_{p = 1}^P {{\mathbf{a}}({\theta _{ip}}){e^{j\varphi_{ip}}}},
\end{equation}
where $P$ is the arbitrary number of i.i.d. paths, $\beta_i$ denotes the path loss for channel $\mathbf{h}_i$, and $e^{j\varphi_{ip}}$ is the i.i.d. random phase, which is independent over channel index $i$ and path index $p$. ${\mathbf{a}}({\theta})$ is the signature (or phase response) vector by the array to a path originating from the angle $\theta $. Note that in the case of an equi-spaced array, ${\mathbf{a}}(\theta)$ has a Fourier structure.

\subsection{Channel Estimation}

When it comes to channel estimation it is assumed that orthogonal pilots are used by users located in the same cell, so that intra-cell pilot interference can be neglected. Sets of pilot sequences are however assumed to be fully reused from cell to cell, causing maximum inter-cell pilot interference.
The pilot sequence is denoted by:
 \begin{equation}\label{Eq:pilots}
    {\mathbf{s}} = {[\begin{array}{*{20}{c}}
  {{s_{1}}}&{{s_{2}}}& \cdots &{{s_{\tau }}}
\end{array}]^T}.
 \end{equation}
The power of the pilot sequence is assumed to be $\mathbf{s}^H \mathbf{s} = \tau$.
The channel vector between the $b$-th cell user and the target base station is ${\mathbf{h}}_{b}$. Without loss of generality, we assume the 1st cell is the target cell. Thus, ${\mathbf{h}}_{1}$ is the desired channel while ${\mathbf{h}}_{b}$, $b>1$ are interference channels.
During the pilot phase, the signal received at the target base station is
\begin{equation}\label{Eq:train}
{{\mathbf{Y}}} = \sum\limits_{b = 1}^B {{{\mathbf{h}}_{b}}{\mathbf{s}}^T + {{\mathbf{N}}}},
\end{equation}
where ${{\mathbf{N}}} \in {\mathbb{C}^{M \times \tau }}$ is the spatially and temporally white additive  Gaussian noise (AWGN) with zero-mean and element-wise variance ${\sigma _n^2}$.
Assuming the desired and interference covariance matrices ${{\mathbf{R}}_{b}} \triangleq \mathbb{E}\{ {{{\mathbf{h}}_{b}}{\mathbf{h}}_{b}^H} \}$ can be estimated in a preamble, the
Bayesian (or equivalently MMSE) estimate of the target channel vector is given by \cite{bjornson2010, hoydis2013, yin2012jsac}:
\begin{align}
{{\mathbf{\widehat h}}_{1}} = {{\mathbf{R}}_{1}}{\left( {\sigma _n^2{{\mathbf{I}}_M} + \tau \sum\limits_{b = 1}^B {{{\mathbf{R}}_{b}}} } \right)^{ - 1}}{{\mathbf{\bar S}}^H}{\mathbf{y}} \label{Eq:EstimatorDesired},
\end{align}
where the training matrix ${\mathbf{\bar S}} \triangleq {\mathbf{s}} \otimes {{\mathbf{I}}_M}$ and ${\mathbf{y}} \triangleq \mbox{vec}(\mathbf{Y})$.
An interesting question is under which conditions $ {{\mathbf{\widehat h}}_1} \rightarrow {\mathbf{\widehat h}}_1^{{\text{no int}}}$ in the massive MIMO regime ($M \gg 1$),
where the superscript \emph{no int} refers to the ``no interference case." This question was previously addressed in \cite{yin2012jsac}, revealing the following sufficient condition for achieving total interference suppression in the large $M$ regime:
\begin{equation}
\mathop  \cup \limits_{b = 2}^B  \mbox{span}  \left\{ {\mathbf{R}}_{b}\right\} \subset \mbox{null}  \left\{ {\mathbf{R}}_{1}\right\} \label{condition}
\end{equation}
where the above condition requires the target channel covariance to exhibit a non-empty null space (aka low-dimensional subspace) {\em and} for all other interference covariances' signal subspaces to fall within this null space (see the proof in \cite{yin2012jsac}). In practice, the inclusion condition in (\ref{condition}) can be realized by a user grouping algorithm \cite{yin2012jsac, caire2012joint}, as long as the rank of each covariance is small enough in relation to $M$.

\subsection{Low-rank Properties of General Linear Arrays}\label{Sec:LowDim_Centralized}

In \cite{yin2012jsac} \cite{caire2012joint}, a linear equi-spaced array was considered. The propagation model also assumed that multipaths impinge on the base station array with angles of arrival (AOA) spanning an interval
$[\theta^{\min }, \theta^{\max }] \in [0,\pi]$ \footnote{Note that a path coming from angle $-\theta$ yields identical steering vector to that from $\theta$. Therefore we can limit ourselves to AOAs within $[0,\pi]$.}. It is then shown that condition (\ref{condition}) is satisfied provided AOAs corresponding to interfering users fall {\em outside} $[\theta^{\min }, \theta^{\max }]$. The assumptions of a single cluster of multipath and of a calibrated equi-spaced array are however restrictive. Below, we generalize this result to more realistic settings.

\subsubsection{Multiple Clusters}\label{Sec:MultiClusterAOA}

We now consider a general multipath model when the AOAs corresponding to the desired channel are still bounded, but come from several disjoint clusters \cite{molisch2010}. The steering vector in (\ref{Eq:simpleChanModel}) is \cite{tsai2002}
\begin{equation}\label{Eq:steeVec}
{\mathbf{a}}({\theta }) \triangleq \left[ {\begin{array}{*{20}{c}}
  1 \\
  {{e^{ - j2\pi \frac{D}{\lambda }\cos ({\theta})}}} \\
   \vdots  \\
  {{e^{ - j2\pi \frac{{(M - 1)D}}{\lambda }\cos ({\theta})}}}
\end{array}} \right],
\end{equation}
where $D$ is the antenna spacing and $\lambda$ is the signal wavelength.
Let $Q$ denote the number of clusters. Let $[\theta _q^{\text{min}}, \theta _q^{\text{max}}]$ denote the interval of AOAs for the $q$-th cluster of desired paths in the $[0,\pi]$ interval. See an illustration in Fig. \ref{fig:InterleavedAOA} for $Q=2$.
\begin{figure}[h]
  \centering
  \includegraphics[width=2.5in]{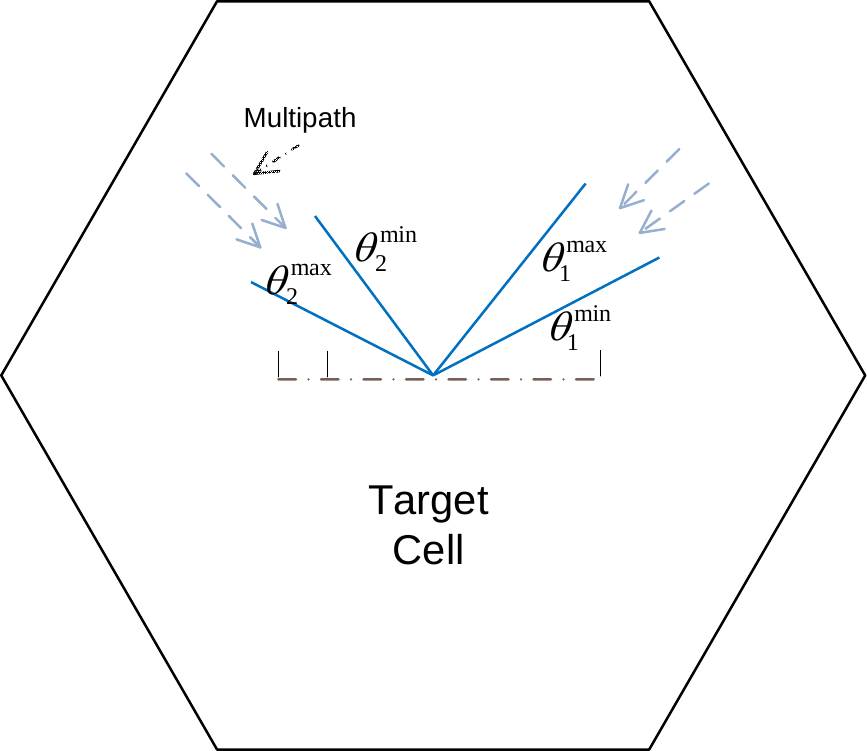}\\
  \caption{Desired channel composed of $Q=2$ clusters of multipath.}  \label{fig:InterleavedAOA}
\end{figure}

For a uniform linear array, we have the following proposition in the massive MIMO regime:
\begin{proposition}\label{propRank_Interleaved}
The rank of channel covariance matrix ${{\mathbf{R}}}$ satisfies:
\begin{equation*}
\frac{\operatorname{rank} ({{\mathbf{R}}})} {M} \leqslant {d}, \text{ when }M \text{ is sufficiently large },
\end{equation*}
where ${d}$ is defined as
\begin{equation*}
{d} \triangleq  \min (1,\sum\limits_{q = 1}^Q {\left( {\cos (\theta _q^{{\text{min}}}) - \cos (\theta _q^{{\text{max}}})} \right)} \frac{D}{\lambda }).
\end{equation*}
\end{proposition}

\begin{proof}
\quad \emph{Proof:} The channel can be seen as the sum of elementary channels each of which corresponds to one separate clusters. Then ${\mathbf{R}}$ can be decomposed into a sum of covariances over these clusters. Since the clusters are separated, the signal subspaces of the corresponding covariances are orthogonal and therefore their dimensions add up. Then based on \cite{yin2012jsac} Lemma 1, the proof of Proposition \ref{propRank_Interleaved} can be readily obtained.
\end{proof}

Now define the total set of AOAs of the desired channel as
\begin{equation}\label{Eq:AOASet}
\overline{\theta}_d \triangleq \cup_{q=1}^Q [\theta _q^{\text{min}}, \theta _q^{\text{max}}],
\end{equation}
so that the probability density function (PDF) $p_d(\theta)$ of the desired AOA satisfies $p_d(\theta)>0$ if $\theta \in \overline{\theta}_d$ and $p_d(\theta)=0$ if $\theta \notin \overline{\theta}_d$. In the same way, the PDF of all interference AOAs satisfies $p_i(\theta)>0$ if $\theta \in \overline{\theta}_i$ and $p_i(\theta)=0$ otherwise, where $\overline{\theta}_i$ is the union of all possible interference AOAs. We have the following result for the massive uniform array:
\begin{corollary}\label{theoremInterleavedAOA}
if $D \leq \lambda/2 $ and $\overline{\theta}_d \cap \overline{\theta}_i = \emptyset$, then the MMSE estimate of (\ref{Eq:EstimatorDesired}) satisfies:
\begin{equation}\label{Eq:EstAsy}
\mathop {\lim }\limits_{M \to \infty } {{\mathbf{\widehat h}}_1} = {\mathbf{\widehat h}}_1^{{\text{no int}}}.
\end{equation}
\end{corollary}
\begin{proof}
\quad \emph{Proof:} It can be shown that from \cite{yin2012jsac} Lemma 2, condition (\ref{condition}) will be fulfilled as long as interfering AOAs do not overlap with {\em any} of the clusters for the desired channel, in which case if we analyze the received signal using eigen-value decomposition, we can find the interference disappears asymptotically because of its orthogonality with the signal space of desired channel covariance. (\ref{Eq:EstAsy}) is obtained in the same way as \cite{yin2012jsac}. As a result we omit the detailed proof in this paper.
\end{proof}


\subsubsection{Random Arrays}

Tightly calibrated arrays with uniform spacing are hard to realize in practice. An interesting question is whether the above results carry on to the setting of linear arrays with random antenna placement. To study this case, we consider a set of antennas randomly located over a line, and spanning a total aperture of $\mathcal{D}$ meters. We investigate the extended array and $\mathcal{D}$ is allowed to grow with $M$.

In this case, an elementary path coming from an angle $\theta$ can be represented via the corresponding array response vector as:
\begin{equation}\label{Eq:steeVecNonCal}
{\mathbf{a}}({\theta }) \triangleq \left[ {\begin{array}{*{20}{c}}
  {{e^{ - j2\pi \frac{d_1}{\lambda }\cos ({\theta})}}} \\
   \vdots  \\
  {{e^{ - j2\pi \frac{d_M}{\lambda }\cos ({\theta})}}}
\end{array}} \right],
\end{equation}
where the position of the $m$-th antenna\footnote{Note that antenna ordering has no impact on our results.}  ($1 \leq {m} \leq M$), $d_m$, follows a uniform distribution, i.e., $d_m \sim \mathcal{U} (0, \mathcal{D})$. The PDF of AOA $\theta$ for the desired paths is non-zero only when $\theta \in \overline{\theta}_d$, as in section \ref{Sec:MultiClusterAOA}.
Define the average antenna spacing $\overline{D} \triangleq \mathcal{D}/M$. Assuming the aperture of antenna array $\mathcal{D}$ is increasing linearly with $M$, i.e., $\overline{D}$ is constant, we now have the extended results on the low-dimensional property:
\begin{proposition}\label{propDim_nonCali}
Define
\begin{align*}
\boldsymbol{\alpha }(x) & \triangleq \left[{{e^{ - j2\pi \frac{d_1}{\lambda }x}}}, \cdots,  {{e^{ - j2\pi \frac{d_M}{\lambda }x}}}\right]^T \\
\mathcal{B} & \triangleq \operatorname{span} \{ \boldsymbol{\alpha }(x), x \in [{b_1},{b_2}]\} \\
\mathcal{C} & \triangleq \operatorname{span} \{ {\boldsymbol{\alpha }}(x), x \in \overline{b}\},
\end{align*}
where  ${b_1},{b_2} \in [ - 1,1]$, $\overline{b} \triangleq \cup_{q=1}^Q [b _q^{\text{min}}, b _q^{\text{max}}]$, and $b _q^{\text{min}}$, $b _q^{\text{max}}$ are values such that
$$ - 1 \leqslant b_1^{{\text{min}}} < b_1^{{\text{max}}} <  \cdots  < b_q^{{\text{min}}} < b_q^{{\text{max}}} <  \cdots  < b_Q^{{\text{min}}} < b_Q^{{\text{max}}} \leqslant 1$$
then we have
\begin{itemize}
  \item $\dim \{ \mathcal{B}\}  \leq {({b_2} - {b_1})M \overline{D}/\lambda} + o(M)$
  \item $\dim \{ \mathcal{C}\}  \leq  \sum\nolimits_{q = 1}^Q {\left( {b _q^{{\text{max}}} - b _q^{{\text{min}}}} \right)M \overline{D}/\lambda} + o(M)$
\end{itemize}
\end{proposition}
\begin{proof}
\quad \emph{Proof:}
See Appendix \ref{proof:propDim_nonCali}.
\end{proof}
Proposition \ref{propDim_nonCali} indicates the dimensions spanned in massive MIMO regime by elementary paths for (i) single cluster of AOA, and (ii) multiple disjoint clusters of AOA, respectively. The following result now directly generalizes Proposition \ref{propRank_Interleaved} to random arrays.

\begin{proposition}\label{propRank_NonCali}
With a bounded support of AOAs $\overline{\theta}_d $ as in (\ref{Eq:AOASet}), the rank of channel covariance matrix ${{\mathbf{R}}}$ satisfies:
\begin{equation}
{\operatorname{rank} ({{\mathbf{R}}})} \leq {\sum\limits_{q = 1}^Q {\left( {\cos (\theta _q^{{\text{min}}}) - \cos (\theta _q^{{\text{max}}})} \right)} \frac{M \overline{D}}{\lambda }} + o(M),
\end{equation}
\end{proposition}
\begin{proof}
\quad \emph{Proof:} We can readily obtain this result by replacing $x$ with $\cos(\theta)$ in Proposition \ref{propDim_nonCali}.
\end{proof}
This result above suggests that the low-dimensional feature of signal subspaces in massive MIMO is not critically linked to the Fourier structure of the steering vectors. Furthermore, it should be noted that the above upper bound is actually very tight for large $M$, as witnessed from the simulation in Fig. \ref{fig:RankLinearRandom}, where we take $Q=1, \overline{D} = \lambda/2$ for example. The AOA spread is 40 degrees, and the closed form model refers to
$$f(M) \triangleq {\sum\limits_{q = 1}^Q {\left( {\cos (\theta _q^{{\text{min}}}) - \cos (\theta _q^{{\text{max}}})} \right)} \frac{M \overline{D}}{\lambda }}.$$
We can observe that ${\operatorname{rank} ({{\mathbf{R}}})}$ is well approximated by $f(M)$.
\begin{figure}[h]
  \centering
  \includegraphics[width=3.2in]{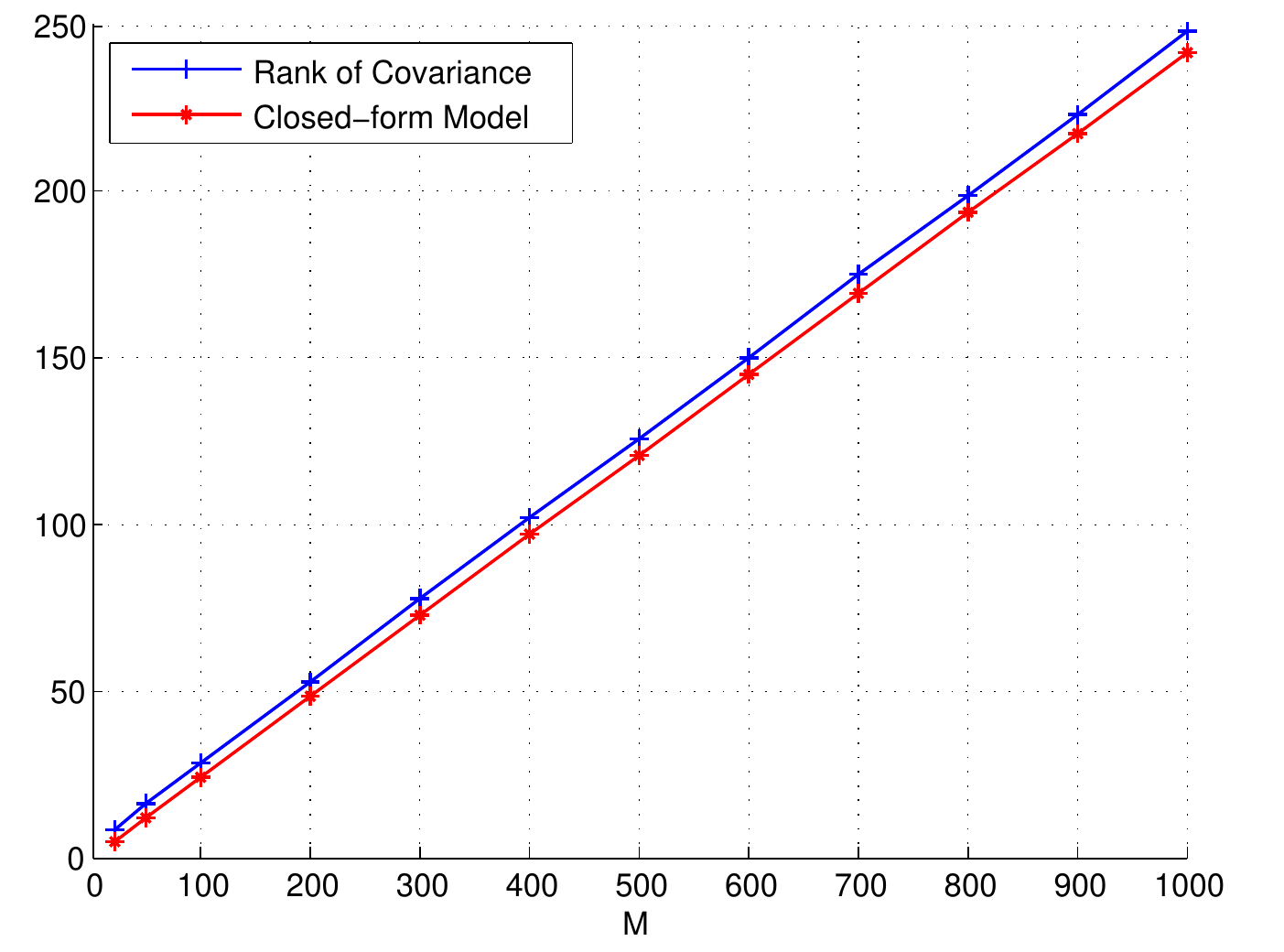}\\
  \caption{Closed-form rank model for the channel covariance vs. actual rank.} \label{fig:RankLinearRandom}
\end{figure}

Proposition \ref{propRank_NonCali} and Fig. \ref{fig:RankLinearRandom} suggest that a property of rank additivity holds for multiple disjoint clusters of AOAs in the massive MIMO regime, i.e. for $M\rightarrow +\infty.$ In the following proposition we extend the results of Corollary \ref{theoremInterleavedAOA} to the case of random arrays under the weaker assumption of rank additivity for the covariance matrices of the desired and interference channels.

\begin{proposition}\label{prop_rank_additivity_error_covariance}
Let $\mathbf{R}_d$ be the covariance matrix of desired channel and $\mathbf{R}_i$ be the covariance of the sum of all interference channels. If $\mathbf{R}_d$ and $\mathbf{R}_i$ satisfy the following rank additivity property
\begin{equation*}
    \mathrm{rank} (\mathbf{R}_d+\mathbf{R}_i) = \mathrm{rank} (\mathbf{R}_d) + \mathrm{rank} (\mathbf{R}_i),
\end{equation*}
then in the high SNR regime, the linear MMSE estimate of the desired channel is error free, or, in other words, its error covariance matrix $\mathbf{C}_e$ vanishes.
\end{proposition}
\begin{proof}
\quad \emph{Proof:} In the case of absence of white Gaussian noise, i.e. $\sigma_n^2=0,$ and rank deficient signal and interference covariance matrices, the error covariance matrix of linear MMSE estimator \cite{Kay1993} can be generalized as
\begin{equation}\label{eq:general_error_covaraince}
    \mathbf{C}_e= \mathbf{R}_d -\mathbf{R}_d (\mathbf{R}_d+\mathbf{R}_i)^{\dag}\mathbf{R}_d
\end{equation}
where $(\cdot)^{\dag}$ denotes the Moore-Penrose generalized inverse of the matrix argument. Let us denote by $\mathbf{R}_{x}=\mathbf{U}_x \boldsymbol{\Sigma}_x \mathbf{U}_x^H,$ with $x \in \{ d,e \},$ $\mathbf{U}_x$ unitary matrix, and $\boldsymbol{\Sigma}_x$ diagonal matrix, the eigenvalue decomposition of the Hermitian matrix $\mathbf{R}_x.$ Then, $\mathbf{R}^{\dag}_x=\mathbf{U}_x \boldsymbol{\Sigma}_x^{\dag} \mathbf{U}_x^H,$ where the elements $i,j$ of the matrix $\boldsymbol{\Sigma}_x^{\dag}$ are given by
\begin{equation*}
    \boldsymbol{\Sigma}_{x,ij}^{\dag}= \left\{
                                         \begin{array}{ll}
                                           \boldsymbol{\Sigma}_{x,ii}^{-1}, & \hbox{if $i=j$ and $\boldsymbol{\Sigma}_{x,ii}\neq 0$;} \\
                                           0, & \hbox{otherwise.}
                                         \end{array}
                                       \right.
\end{equation*}
Additionally, $\widetilde{\mathbf{U}}_x$ denotes the column space of $\mathbf{R}_x$ and $\widetilde{\boldsymbol{\Sigma}}_x$ the corresponding nonzero eigenvalues such that $\mathbf{R}_x=\widetilde{\mathbf{U}}_x\widetilde{\boldsymbol{\Sigma}}_x \widetilde{\mathbf{U}}_x^H.$ Then, under the assumption of rank additivity of the covariance matrices $\mathbf{R}_d$ and $\mathbf{R}_i$, the theorem on the Moore-Penrose generalized inverse for sum of matrices in \cite{fill:00} yields
\begin{equation}\label{eq:pseudoinverse_sum}
    (\mathbf{R}_d+\mathbf{R}_i)^{\dag} = (\mathbf{I}-\mathbf{S}^{\dag}) \mathbf{R}_d^{\dag} (\mathbf{I}-\mathbf{T}^{\dag}) +\mathbf{S}^{\dag} \mathbf{R}_i^{\dag} \mathbf{T}^{\dag},
\end{equation}
where $\mathbf{S}=\widetilde{\mathbf{U}}_i\widetilde{\mathbf{U}}_i^H(\mathbf{I}-\widetilde{\mathbf{U}}_d\widetilde{\mathbf{U}}_d^H)$ and $\mathbf{T}=(\mathbf{I}-\widetilde{\mathbf{U}}_d\widetilde{\mathbf{U}}_d^H) \widetilde{\mathbf{U}}_i\widetilde{\mathbf{U}}_i^H.$

Let us observe that
\begin{equation}\label{eq:orthogonality_propreties}
    \mathbf{T}^{\dag} \mathbf{R}_d =\mathbf{0}  \quad \text{ and } \quad \mathbf{R}_d \mathbf{S}^{\dag}=\mathbf{0}.
\end{equation}
We focus on the first equality. The proof of the second equation follows along the same line. By appealing to the mixed type reverse order laws of the $r\times s$ matrix $\mathbf{A}$ and the $s\times t$ matrix $\mathbf{B}$    in \cite{rakha:04}
\begin{equation*}
    (\mathbf{A} \mathbf{B})^{\dag} =\mathbf{B}^H (\mathbf{A}^H\mathbf{A}\mathbf{B}\mathbf{B}^H)^{\dag} \mathbf{A}^H,
\end{equation*}
$\mathbf{T}^{\dag} $ can be rewritten as
\begin{align}\label{eq:pseudoinverse_product}
    \mathbf{T}^{\dag} &=  \widetilde{\mathbf{U}}_i\widetilde{\mathbf{U}}_i^H \left[(\mathbf{I}-\widetilde{\mathbf{U}}_d\widetilde{\mathbf{U}}_d^H) \widetilde{\mathbf{U}}_i\widetilde{\mathbf{U}}_i^H \right]^{\dag} (\mathbf{I}-\widetilde{\mathbf{U}}_d\widetilde{\mathbf{U}}_d^H) \nonumber \\
    &= \widetilde{\mathbf{U}}_i\widetilde{\mathbf{U}}_i^H \mathbf{T}^{\dag} (\mathbf{I}-\widetilde{\mathbf{U}}_d\widetilde{\mathbf{U}}_d^H). \nonumber
\end{align}
The first equality is obtained utilizing the fact that the  matrices  $\widetilde{\mathbf{U}}_i\widetilde{\mathbf{U}}_i^H $ and $ (\mathbf{I}-\widetilde{\mathbf{U}}_d\widetilde{\mathbf{U}}_d^H)$ are orthogonal projectors and thus idempotent. Then,
\begin{equation*}
    \mathbf{T}^{\dag} \mathbf{R}_d= \widetilde{\mathbf{U}}_i\widetilde{\mathbf{U}}_i^H  \mathbf{T}^{\dag} (\mathbf{I}-\widetilde{\mathbf{U}}_d\widetilde{\mathbf{U}}_d^H) \widetilde{\mathbf{U}}_d \widetilde{\boldsymbol{\Sigma}}_d \widetilde{\mathbf{U}}_d^H = \mathbf{0}.
\end{equation*}
Finally, substituting (\ref{eq:pseudoinverse_sum}) into (\ref{eq:general_error_covaraince}) and accounting for  orthogonality in (\ref{eq:orthogonality_propreties})
\begin{align}
    \mathbf{C}_e &=\mathbf{R}_d - \mathbf{R}_d \left[(\mathbf{I}-\mathbf{S}^{\dag}) \mathbf{R}_d^{\dag} (\mathbf{I}-\mathbf{T}^{\dag})  +\mathbf{S}^{\dag} \mathbf{R}_i^{\dag} \mathbf{T}^{\dag} \right] \mathbf{R}_d \nonumber \\
    &= \mathbf{R}_d - \mathbf{R}_d \mathbf{R}_d^{\dag} \mathbf{R}_d = \mathbf{0}. \nonumber
\end{align}
In the last equality we use one of the fundamental relations defining the Moore-Penrose generalized inverse.
\end{proof}

According to Proposition \ref{propRank_NonCali}, the rank additivity condition is in general satisfied when the AOA support of desired channel and that of interference channels span disjoint region of spaces, i.e., $\overline{\theta}_d \cap \overline{\theta}_i = \emptyset$. This property can be exploited in pilot decontamination or interference rejection.
Fig. \ref{fig:ChanEstRandomArray} shows the channel estimation performance in the presence of contaminating pilots. In the simulation, we consider a 2-cell network. Each cell has one single-antenna user who uses identical pilot sequence. The mean squared error (MSE) of uplink channel estimation is shown. The simulation suggests that the MMSE channel estimator is able to rid itself from pilot contamination effects as the number of antennas is (even moderately) large, which verifies Proposition \ref{prop_rank_additivity_error_covariance}.
\begin{figure}[h]
  \centering
  \includegraphics[width=3.2in]{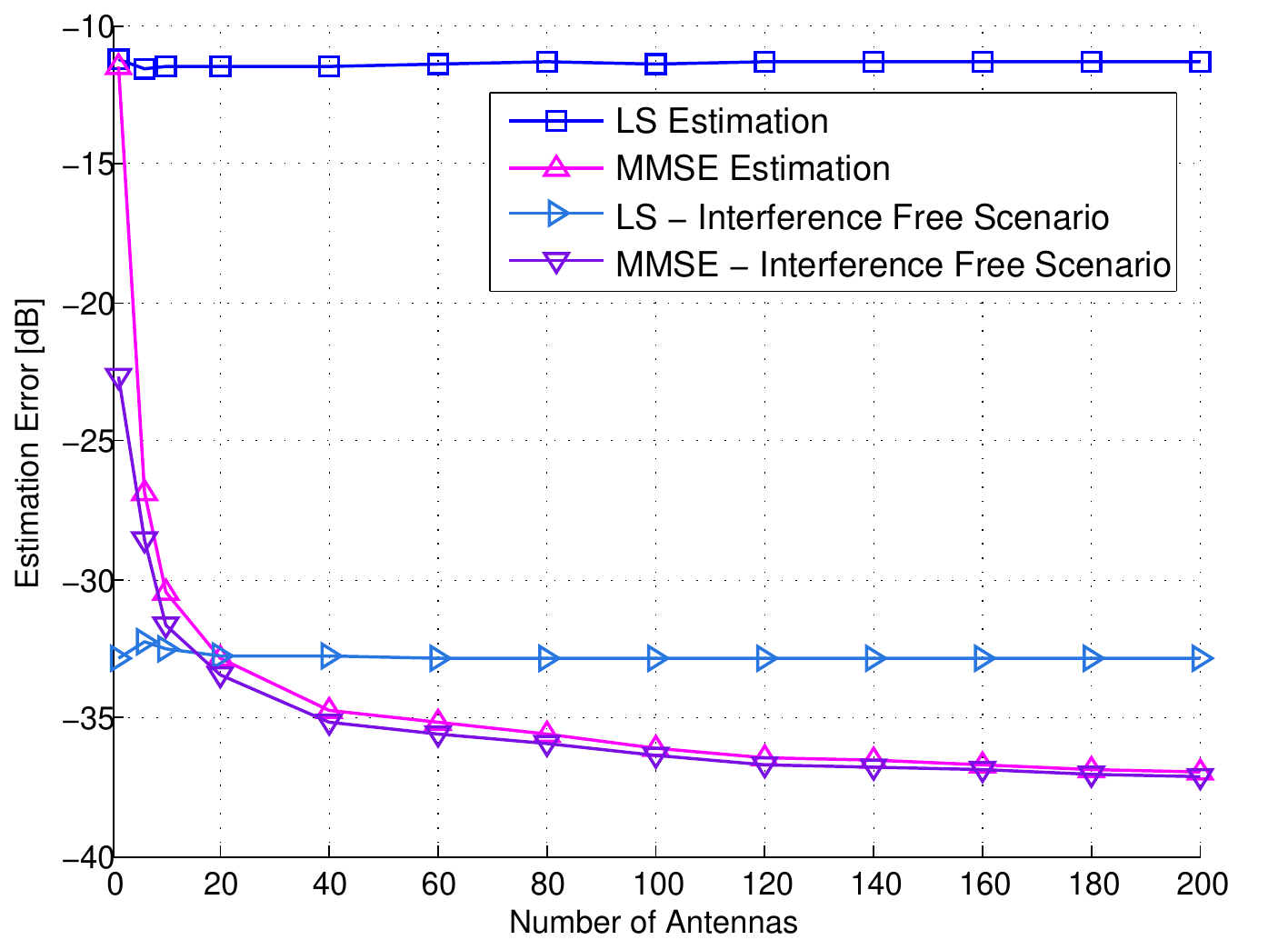}\\
  \caption{Channel estimation performance vs. $M$, $\overline{D} = \lambda/2$, 2-cell network, angle spread 30 degrees, $\overline{\theta}_d \cap \overline{\theta}_i = \emptyset$, cell-edge SNR is 20dB. We compare the standard Least Squares (LS) to MMSE estimators, in interference and interference-free scenarios.} \label{fig:ChanEstRandomArray}
\end{figure}

\section{Finite Rank Model in Distributed Arrays}
We now turn to another popular form of large scale antenna regime, often referred to in the literature as distributed antenna systems. In such a setting, a virtual base station is deployed having its $M$ antennas scattered throughout the cell.\footnote{For ease of exposition we temporarily consider a single cell setting in this section, i.e., $B=1$. However simulation is also done later in a multi-cell scenario.} We consider again the uplink in which joint combining across all BS antennas is assumed possible.
The $M$ base station antennas are assumed uniformly and randomly located in a fixed size network, serving single-antenna users. $M$ is allowed to grow large giving rise to a so-called dense network. Our model assumes a disk-shaped cell of radius $L$, although simulation and intuition confirm that the actual shape of the cell's boundary is irrelevant to the main result.

\subsection{Channel Model}
In order to facilitate the analysis, we adopt the one-ring model \cite{jakes1974, Shiu2000} where users are surrounded by a ring of $P$ local scatterers (see Fig. \ref{fig:DistributedMassiveMIMO}) located $r$ meters away from the user. The positions of the scatterers are considered to follow a uniform distribution on the ring. In the one-ring model, the propagation from user to base is assumed to follow $P$ paths (hereafter referred to as scattering paths), where each path $p$ bounces once on the $p$-th scatterer before reaching all $M$ destinations.\footnote{Note that this model assumes the BS antennas are high enough above clutter so that there is no local scattering around the BS antennas.}

Hence, the path length from user $k$ to the $m$-th antenna via the $p$-th path is $r + d_{kpm}$, where $d_{kpm}$ is the distance between the $p$-th scatterer of the $k$-th user and the $m$-th BS antenna. The path loss of the $p$-th scattering path is modeled by:
\begin{equation}\label{Eq:pathloss}
{\beta _{kpm}} = \frac{\alpha }{{{(d_{kpm}+r)}^\gamma }},
\end{equation}
where $\alpha$ is a constant that can be computed based on desired cell-edge SNR, and $\gamma$ is the path loss exponent. We scale the amplitude of each path by $\sqrt P$. The channel between user $k$ and all BS antennas is given by:
\begin{equation}\label{Eq:chanDistri}
{\mathbf{h}_k} \triangleq \frac{1}{\sqrt{P}}\sum\limits_{p = 1}^P {{{\mathbf{h}}_{kp}}},
\end{equation}
where ${\mathbf{h}_{kp}}$ is the $p$-th scattering path vector channel between user $k$ and all base stations:
\begin{equation}\label{Eq:chanPth}
{\mathbf{h}_{kp}} \triangleq  \left[ {\begin{array}{*{20}{c}}
  {\sqrt{\beta _{kp1}}{e^{ - j2\pi \frac{d_{kp1}+r}{\lambda }}}} \\
   \vdots  \\
  {\sqrt{\beta _{kpM}}{e^{ - j2\pi \frac{d_{kpM}+r}{\lambda }}}}
\end{array}} \right] e^{j\varphi_{kp}},
\end{equation}
where $e^{j\varphi_{kp}}$ denotes the random common phase of that scattering path vector due to possible random perturbations of the user location around the ring center or the phase shift due to the reflection on the scatterer. $\varphi_{kp}$ is assumed i.i.d. and uniformly distributed between 0 to $2\pi$.

\begin{figure}[h]
  \centering
  \includegraphics[width=3.2in]{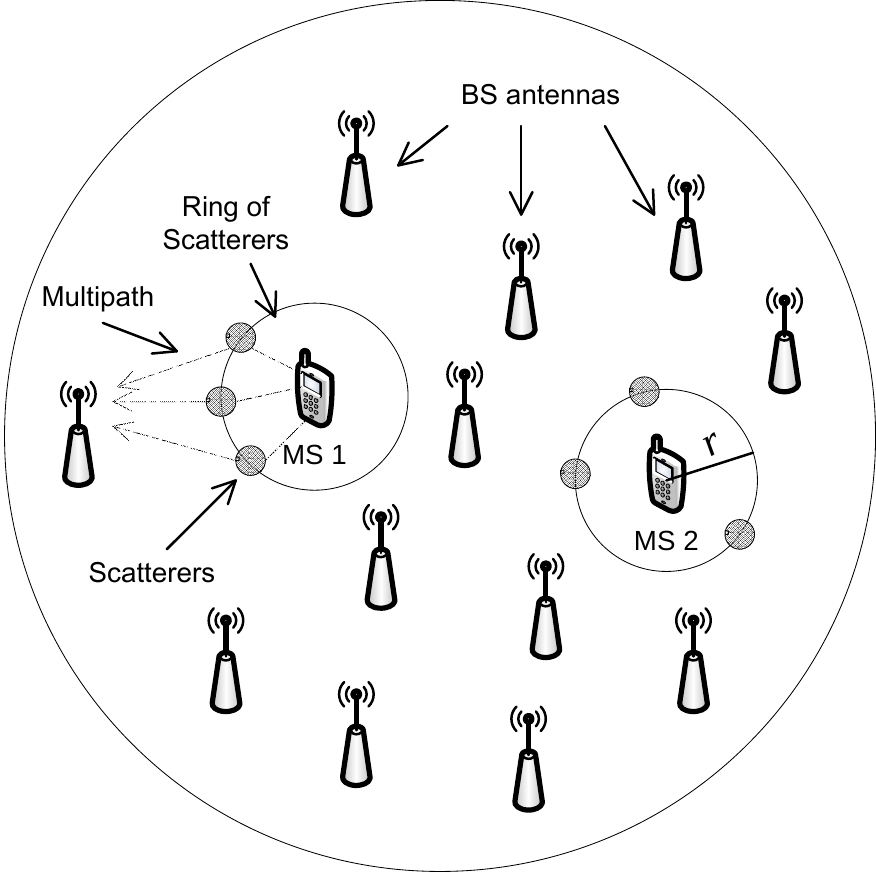}\\
  \caption{The distributed large-scale antenna setting with a one-ring model.} \label{fig:DistributedMassiveMIMO}
\end{figure}

\subsection{A Low-rank Model for Distributed Arrays}\label{Sec:LowRankDAS}
We have shown in section \ref{Sec:LowDim_Centralized} the low-dimension property for linear antenna array systems. In attacking this problem it is important to distinguish the rank reduction effect due to path loss from the intrinsic finite-rank behavior of the large antenna channel covariance in an equal path loss regime. In fact, in an extended network (i.e. where some base station antennas can be arbitrarily far from some users), the signal of any given user will be received over only a limited number of antennas in its vicinity, thereby effectively limiting the channel rank to the size of this neighborhood. To circumvent this problem, we consider below a (dense) network where the path loss terms are set artificially to be all equal (to one) and study the finite-rankness under such conditions. In this model, the channel covariance is defined as ${{\mathbf{R}}} \triangleq \mathbb{E}\{ {{{\mathbf{h}}}{\mathbf{h}}^H} \}$ where the expectation is taken over the random positions of the scatterers on the ring. Note that our analysis indicates that a randomization over the user's location inside the scattering's disk would produce an identical upper bound on the rank.

\begin{theorem}\label{theoremRankDAS}
The rank of the channel covariance matrix for a distributed antenna system satisfies:
\begin{equation}\label{Eq:RankDAS}
{\text{rank}({\mathbf{R}})} \leq  \frac{4\pi r}{\lambda} + o(r).
\end{equation}
\end{theorem}
\begin{proof}
\quad \emph{Proof:} See Appendix \ref{proof:theoremRankDAS}.
\end{proof}
In reality we show below that the right hand side of (\ref{Eq:RankDAS}) is a very close approximation of the actual rank, which is defined as the number of eigenvalues of $\mathbf{R}$ which are greater than a prescribed threshold (in our simulations it is taken to be 10e-5).
Theorem \ref{theoremRankDAS} shows a linear dependency of the rank on the size of the scattering ring. 
When $r$ increases, the richer scattering environment expands the dimension of signal space. 
\begin{figure}[h]
  \centering
  \includegraphics[width=3.2in]{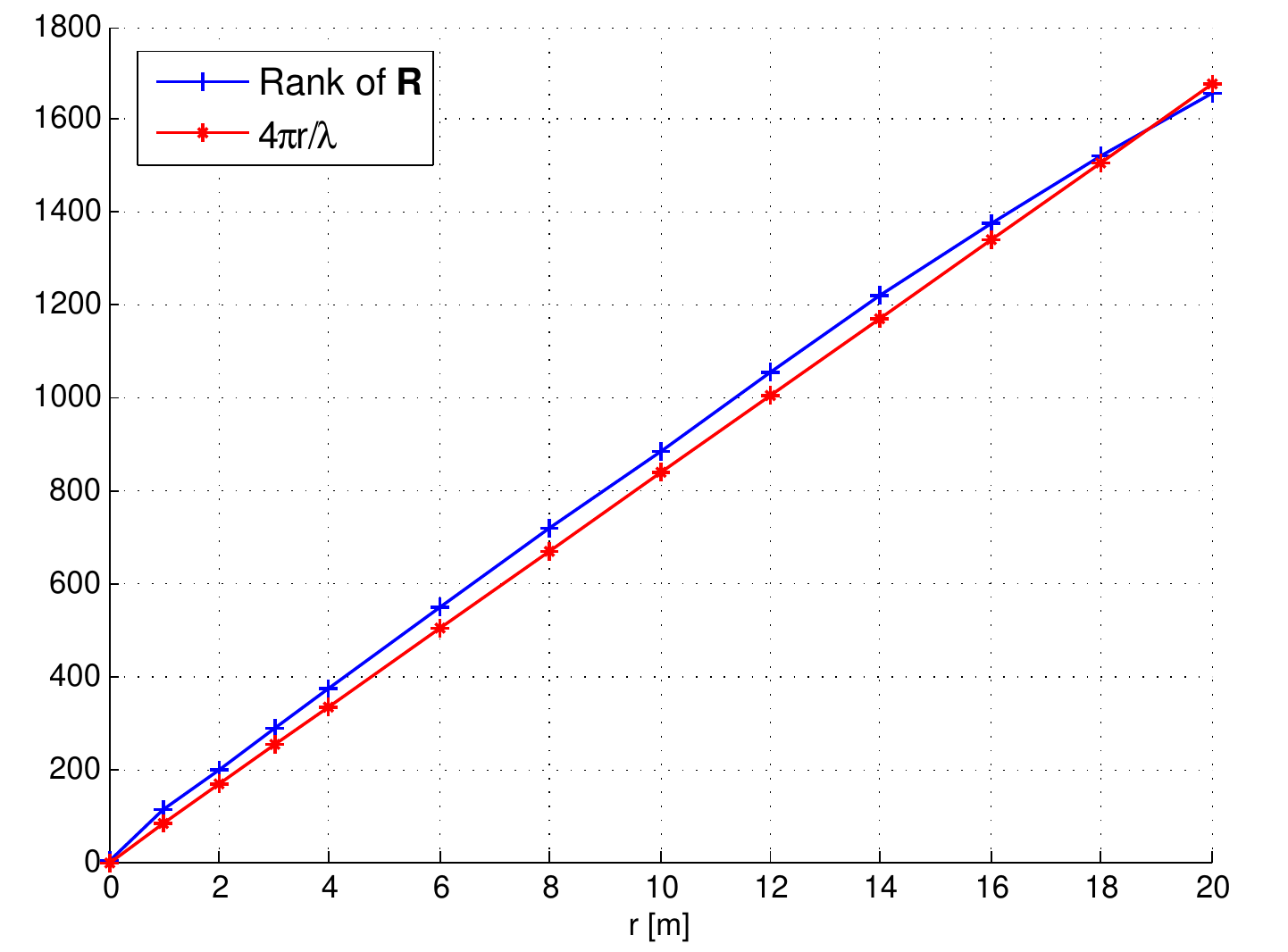}\\
  \caption{Rank vs. $r$, $M = 2000$, $\lambda$ = 0.15m, $L = 500$m.} \label{fig:rank_vs_r}
\end{figure}
Fig. \ref{fig:rank_vs_r} shows the behavior of the covariance rank with respect to the scattering radius $r$. 
We can see the rank scales linearly with the slope $4\pi / \lambda$. However because of the finite number of antennas the rank will finally saturate towards $M$ when $r$ keeps increasing.


\section{Spatial Interference Filtering}

We are now interested in characterizing the orthogonality (or correlation) between any two user channel vectors as a function of inter-user distance, the wavelength and the scattering radius $r$, in the large $M$ limit, as this will provide a measure of interference rejection capability for the distributed antenna systems. In the following we will investigate two interference filtering schemes: 1) the simplistic matched filtering, 2) a subspace projection filtering.

\subsection{Performance of Interference Filtering Using Matched Filter}
We start with analyzing the channel correlation between two users who interfere each other.
We point out two distinct regimes, depending on whether the inter-user distance is small or large.

\subsubsection{Closely Spaced Users}
Closely spaced users are defined by the fact that the distance between user 1's and user 2's scatterers is small enough compared with the distance between scatterers and receiving antennas so that we can consider planar wavefronts. We first examine the correlation between any two scattering paths for user 1 and user 2, corresponding to user 1's $p$-th scatterer and user 2's $q$-th scatterer, with a distance $D_{pq}$.

\begin{proposition}\label{propBessel}
For $D_{pq}$ small enough that the two scatterers are located in the same planar wavefront region, we have
\begin{equation}\label{Eq:Bessel}
\mathop {\lim }\limits_{M \to \infty } \frac{{\left| {{\mathbf{h}}_{2q}^H{{\mathbf{h}}_{1p}}} \right|}}{{\left| {{{\mathbf{h}}_{1p}}} \right|\left| {{{\mathbf{h}}_{2q}}} \right|}} \approx \left| {{J_0}\left(\frac{{2\pi {D_{pq}}}}{\lambda }\right)} \right|,
\end{equation}
where ${J_0}$ is the zero-order Bessel function of the first kind.
\end{proposition}
\begin{proof}
\quad \emph{Proof:} See Appendix \ref{proof:propBessel}.
\end{proof}
Note that the proof is based on an additional assumption that the path loss between a certain antenna and the two scatterers are approximately equal, i.e., if user 1 and user 2 are concerned, then ${\beta _{2qm}} \approx {\beta _{1pm}}$. Since the two scatterers are very close and the antenna is much further away, this assumption is reasonable in practice.
To visualize Proposition \ref{propBessel}, we draw the curves of $\frac{{\left| {{\mathbf{h}}_{2q}^H{{\mathbf{h}}_{1p}}} \right|}}{{\left| {{{\mathbf{h}}_{1p}}} \right|\left| {{{\mathbf{h}}_{2q}}} \right|}}$ and $|{J_0}(\frac{{2\pi {D_{pq}}}}{\lambda })|$ in Fig. \ref{fig:Correlation}.
\begin{figure}[h]
  \centering
  \includegraphics[width=3.2in]{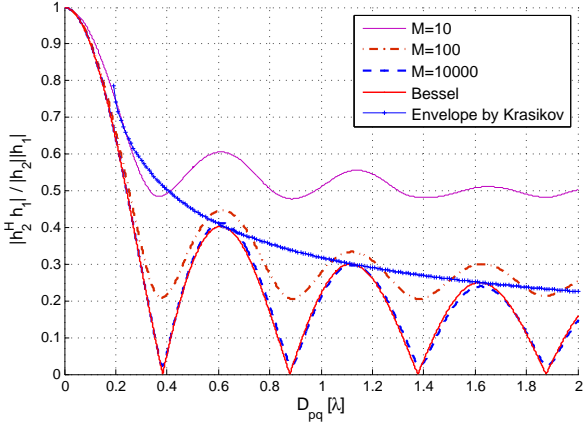}\\
  \caption{Illustration of Proposition \ref{propBessel}.} \label{fig:Correlation}
\end{figure}
The curves show that when $M$ grows, ${\left| {{\mathbf{h}}_{2q}^H{{\mathbf{h}}_{1p}}} \right|}/{\left| {{{\mathbf{h}}_{1p}}} \right|\left| {{{\mathbf{h}}_{2q}}} \right|}$ gets closer and closer to the Bessel function. The curve named ``Envelope by Krasikov" is an upper bound of the Bessel function developed by Krasikov in \cite{krasikov2006}, which will be used in Proposition \ref{propSIRBessel}.
The interpretation of Proposition \ref{propBessel} is as follows: the highest correlation between two scattering paths is attained when the spacing between the two scatterers is very small. Since the Bessel function reaches its first zero around a spacing of $D_0=0.38 \lambda$, it implies that the users ought to be located at least $2r+D_0 \approx 2r$ away from each other to allow for a reuse of spectral resources (such as pilots). In practice we would intend to schedule users with additional spacing than just $2r$, according to the side lobes of the Bessel function.
Note that our model is similar to Clarke's model \cite{clarke1968}, which indicates if the AOA is uniformly distributed from 0 to $2\pi$, the autocorrelation of a moving mobile is a scaled Bessel function.

We use Proposition \ref{propBessel} to derive a lower bound on the signal to interference ratio (SIR) under a simplified system setting with just two users (one desired, one interferer) and a matched filter receiver\footnote{With more users, the interference is simply scaled by the number of users. Additionally more advanced receivers could also be exploited.}. By computing the expectation over the random BS antenna locations, which is also how we derive the expectation of SIRs and channel correlations in the rest of this section, we obtain the following bound of SIR:
\begin{proposition}\label{propSIRBessel}
Assume perfect channel estimation, closely located users, whose scattering rings do not overlap. The expected SIR at the matched filter output satisfies:
$$ \mathbb{E}\{ \text{SIR}\} \gtrsim { {\frac{{\pi \left( {{{\left( {4{{(\frac{{2\pi {(D_{u }-2r)}}}{\lambda })}^2} - 3} \right)}^{\frac{3}{2}}} - 3} \right)}}{{16{{(\frac{{2\pi {(D_{u }-2r)}}}{\lambda })}^2} - 20}}} }, \text{ when } M \text{ is large,}$$
where $D_{u}$ is the distance between the two users and is assumed to be larger than ${{\sqrt {3 + {3^{\frac{2}{3}}}} \lambda }}/{{(4\pi) }} + 2r$.
\end{proposition}
\begin{proof}
\quad \emph{Proof:} When applying matched filter, we have
$$ \mathbb{E}\{ \text{SIR}\} =  \mathbb{E} \frac{{{{\left| {{{\mathbf{h}}_{1}}} \right|}^2}{{\left| {{{\mathbf{h}}_{2}}} \right|}^2}}}{{{{\left| {{\mathbf{h}}_{2}^H{{\mathbf{h}}_{1}}} \right|}^2}}}.$$
Let us recall that the envelope of the Bessel function ${J_0}(\frac{{2\pi D}}{\lambda })$ is decreasing with $D$. Thus, the lower bound of SIR is obtained by considering the shortest $D_{pq}$, which gives the worst case correlation. Since $D_{pq} \geq D_u - 2r$, we may obtain when $M$ is large:
$$\mathbb{E}\{ \text{SIR}\} \gtrsim \frac{1}{\left| {{J_0}(\frac{{2\pi {(D_{u}-2r)}}}{\lambda })} \right|^2}.$$
Finally we use an upper bound of the envelope of the Bessel function \cite{krasikov2006} which has validity when $D_u -2r > {{\sqrt {3 + {3^{\frac{2}{3}}}} \lambda }}/{{(4\pi) }}$. The bounding argument of the Bessel function in \cite{krasikov2006} can directly apply here.
\end{proof}
The above proposition quantifies the rate at which the SIR increases with the inter-user distance, in this case linearly.

\subsubsection{Distant Users}
We consider the regime in which users are located further away from each other, e.g., many wavelengths away. The planar wavefront assumption no longer holds, making the use of the Bessel function impractical.
In this case we are again interested in characterizing the correlation between two scattering paths corresponding to two users, then the correlation between the channel vectors themselves.

We first investigate the behavior of ${{\mathbf{h}}_{2q}}^H{{\mathbf{h}}_{1p}}$ for any $p$, $q$:
\begin{equation}\label{Eq:h2h1}
{{\mathbf{h}}_{2q}}^H{{\mathbf{h}}_{1p}} = \sum\limits_{m = 1}^M {\left( {h_{2qm}^*{h_{1pm}}} \right)} ,
\end{equation}
where $h_{kpm}$ is the channel between the $k$-th user and the $m$-th BS antenna via the $p$-th scatterer:
\begin{equation}
h_{kpm} = \sqrt{\beta _{kpm}}{e^{ - j2\pi \frac{{{d_{kpm}} + r}}{\lambda }}} e^{j\varphi_{kp}}.
\end{equation}
Since the two phases $e^{j\varphi_{1p}}$ and $e^{j\varphi_{2q}}$ are independent, we have:
$$\mathbf{E}(h_{2qm}^*h_{1pm})=0,$$
and the variance of $h_{2qm}^*h_{1pm}$ is
\begin{equation}\label{Eq:Var_h2h1}
\sigma^2(D_{pq}) \triangleq \text{Var}(h_{2qm}^*h_{1pm})=\mathbf{E}(\beta _{2qm}\beta _{1pm}).
\end{equation}
Given the random network model with radius $L$, the path loss correlation can be found by integration over polar coordinates giving the location of the BS antennas. Although a closed-form expression is elusive, we get the following computable expression:
\begin{align}\label{Eq:varianceDerivation}
& \sigma^2(D) = \nonumber \\
& \frac{{2{\alpha}}}{{\pi {L^2}}}\int_0^L {d\rho \int_0^\pi  {\frac{\rho }{{{{\left( {\rho  + r} \right)}^\gamma }{{\left( {\sqrt {D^2 + {\rho ^2} - 2\rho {D}\cos (\varphi )}  + r} \right)}^\gamma }}}} } d\varphi
\end{align}
One example of the path loss correlation $\sigma^2(D)$ is given in Fig. \ref{fig:SigmaSquare}, which shows $\sigma^2(D)$ is a decreasing function of $D$.
\begin{figure}[h]
  \centering
  \includegraphics[width=3.2in]{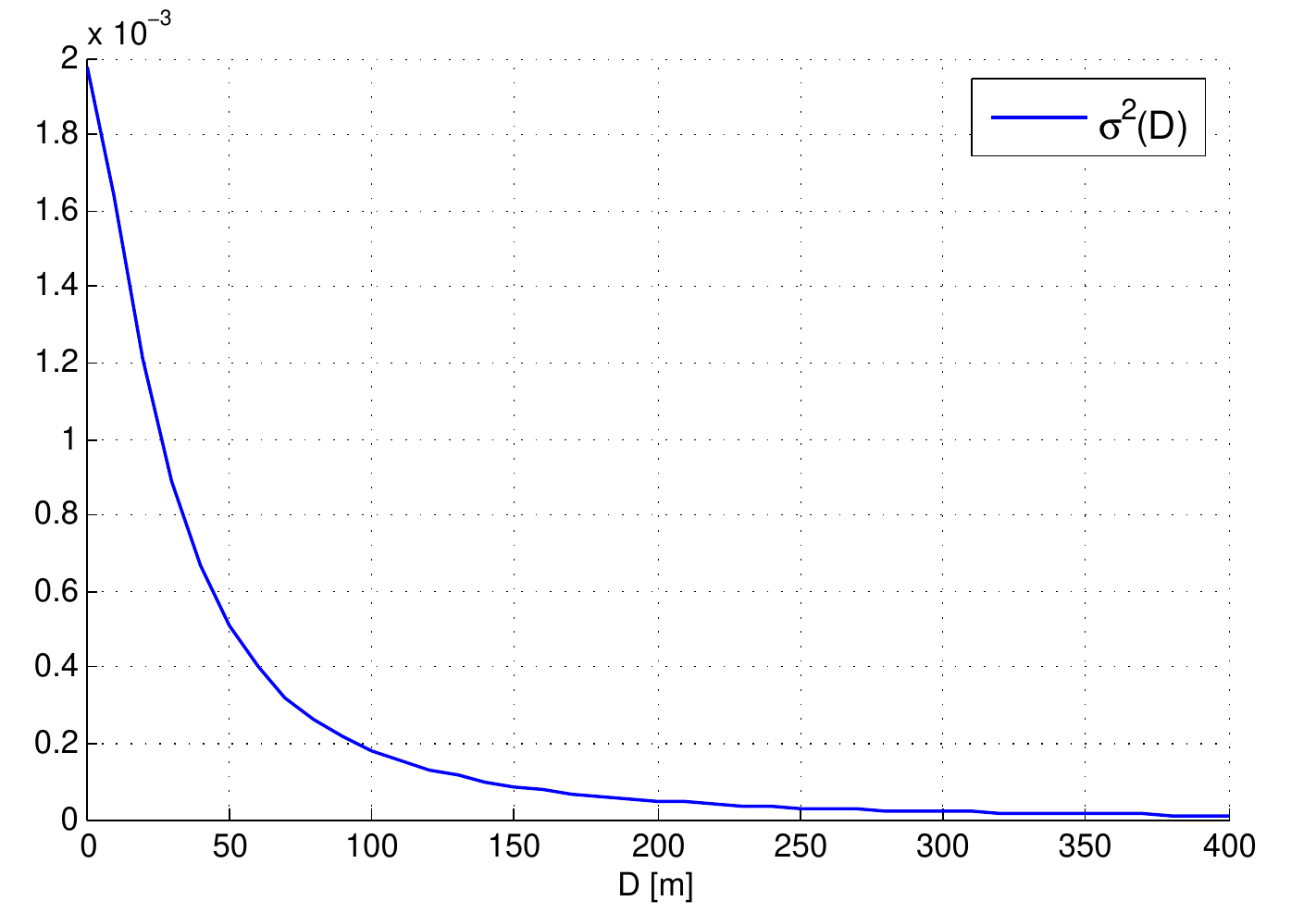}\\
  \caption{An illustration of $\sigma^2(D)$ with $\alpha = 10^7$, $L = 500$, $r = 15$, and $\gamma = 2.5.$} \label{fig:SigmaSquare}
\end{figure}
We have the following proposition on the distribution of $| {{\mathbf{h}}_{2q}^H{{\mathbf{h}}_{1p}}} |^2$:
\begin{proposition}\label{propChi-square}
Let $\chi$ be a random variable exponentially distributed with parameter $\lambda=2$. Asymptotically for $M \rightarrow +\infty$
\begin{equation}
\frac{| {\mathbf{h}}_{2q}^H {\mathbf{h}}_{1p}|^2 }{\sigma^2(D_{pq})M }
\overset{d}{\rightarrow} \chi
\end{equation}
where $\overset{d}{\rightarrow}$ denotes convergence in distribution.
\end{proposition}
\begin{proof}
\quad \emph{Proof:}
Let us consider the random variable $S_M=\frac{ {\mathbf{h}}_{2q}^H {\mathbf{h}}_{1p} }{\sigma(D_{pq})M}. $ By the law of large numbers $S_M$ converges almost surely to $0$ as $M \mapsto +\infty$. By appealing to the Central Limit Theorem (CLT), the random variable $\sqrt{M}S_M$ converges in distribution to a complex Gaussian distribution with zero mean and unit variance. Thus, its square, i.e. $M |S_M|^2$, converges also in distribution to an exponentially distributed random variable with parameter $\lambda=2$.
\end{proof}

We now derive a lower bound of the average SIR for a two-user system, assuming  matched filtering receiver. 
\begin{proposition}\label{propSIR_LB}
The lower bound is given by
\begin{equation}\label{Eq:SIR_LB}
\mathbb{E}\{ \text{SIR}\}  \geqslant \frac{{M(C^2 + o(1))}}{{{\sigma ^2}({D_u} - 2r)}},
\end{equation}
where $C$ is a constant such that ${\mathbf{h}}_1^H{{\mathbf{h}}_1}/M = C + o(1)$.
\end{proposition}
\begin{proof}
\quad \emph{Proof:} The SIR can be written as:
\begin{equation}
\text{SIR} = \frac{{{{\left| {{\mathbf{h}}_1^H{{\mathbf{h}}_1}} \right|}^2}}}{{{{\left| {{\mathbf{h}}_2^H{{\mathbf{h}}_1}} \right|}^2}}} = \frac{{\frac{{{{\left| {{\mathbf{h}}_1^H{{\mathbf{h}}_1}} \right|}^2}}}{{{M^2}}}}}{{\frac{{{{\left| {{\mathbf{h}}_2^H{{\mathbf{h}}_1}} \right|}^2}}}{{{M^2}}}}} = \frac{{{C^2} + o(1)}}{\frac{{{{\left| {{\mathbf{h}}_2^H{{\mathbf{h}}_1}} \right|}^2}}}{{{M^2}}}}.
\end{equation}
Note that $\forall p,q$, ${{\mathbf{h}}_{2q}^H}  {{\mathbf{h}}_{1p}} / M$ has zero mean and a variance of ${{\sigma ^2}({D_{pq}})}/{M}$. In addition, the two variables ${{\mathbf{h}}_{2q}^H}  {{\mathbf{h}}_{1p}} / M$ and ${{\mathbf{h}}_{2q'}^H}  {{\mathbf{h}}_{1p'}} / M$ are uncorrelated for any $p' \neq p$ or $q' \neq q$, resulting from the random and independent phases in (\ref{Eq:chanPth}).
\begin{align*}
\mathbb{E}\left\{ \frac{{{\left| {{\mathbf{h}}_2^H{{\mathbf{h}}_1}} \right|}^2}}{M^2} \right\} &= \mathbb{E}\left\{ {\frac{1}{{{P^2}}}{{\left| {{ \sum\limits_{q = 1}^P {\sum\limits_{p = 1}^P  \frac{{\mathbf{h}}_{2q}^H {{\mathbf{h}}_{1p}}}{M} }}  } \right|}^2}} \right\} \\
 &= \frac{1}{P^2} \text{Var}({{ \sum\limits_{q = 1}^P {\sum\limits_{p = 1}^P  \frac{{\mathbf{h}}_{2q}^H {{\mathbf{h}}_{1p}}}{M} }}  }) \\
 &= \frac{1}{P^2} {{ \sum\limits_{q = 1}^P {\sum\limits_{p = 1}^P  \text{Var} \left( \frac{{\mathbf{h}}_{2q}^H {{\mathbf{h}}_{1p}}}{M}\right) }}  } \\
 &= \frac{1}{P^2} { \sum\limits_{q = 1}^P {\sum\limits_{p = 1}^P  \frac{\sigma^2(D_{pq})}{M} }} \\
 &\leq \frac{\sigma^2(D_u - 2r)}{M}.
\end{align*}
The final step is due to the fact that $D_{pq} \geq D_u - 2r, \forall p,q$. Finally we get
$$\mathbb{E}\{ \text{SIR}\}  \geqslant \frac{{M(C^2 + o(1))}}{{{\sigma ^2}({D_u} - 2r)}},$$
and Proposition \ref{propSIR_LB} is proven.
\end{proof}
The above result suggests at which rate the matched filtered interference decays as a function of $M$ and the inter-user spacing $D_u$, which in turn can be exploited to predict system performance and also give insights to the tolerable spatial reuse of pilot resources.

\subsection{Interference Filtering via Subspace Projection}\label{subspaceBF}
SIR analysis in previous sections is built upon a simple matched filter, which still requires an accurate channel estimation. In this section, however, we propose a simple beamforming strategy building on the low-dimensionality of the signal subspace, which does not require an accurate channel estimation.
We consider a $K$-user network with the first user being a target user and all other users being interference users. All these users share the same pilot sequence $\mathbf{s}$. 
Denote the sum of interference covariances as $\mathbf{R}_I = \mathbf{R}_2 + \mathbf{R}_3 + \cdots + \mathbf{R}_K$.
The eigenvalue decomposition of $\mathbf{R}_I$ is 
${\mathbf{R}_I=\mathbf{U}\mathbf{\Sigma}\mathbf{U}^{H}}$,
where $\mathbf{\Sigma}$ is a $M \times M$ diagonal matrix with the eigenvalues of $\mathbf{R}_I$ on its main diagonal. Suppose the eigenvalues are in descending order and the first $m$ eigenvalues are non-negligible while the others can be neglected.
We construct the spatial filter at the BS side for user 1 as:
\begin{equation}
\mathbf{W}_1 = \left[ \mathbf{u}_{m+1} | \mathbf{u}_{m+2}| \ldots | \mathbf{u}_M\right]^H,
\end{equation}
where $\mathbf{u}_{m}$ is the $m$-th column of $\mathbf{U}$.
We can assume approximately that:
$$\mathbf{W}_1 {\mathbf{h}}_{k} \approx \mathbf{0}, \forall k \ne 1,$$
\begin{equation}\label{Eq:train2Cell2}
{{\mathbf{W}_1\mathbf{Y}}} \approx \mathbf{W}_1{{\mathbf{h}}_{1}}{\mathbf{s}}^T + \mathbf{W}_1 {{\mathbf{N}}},
\end{equation}
where ${{\mathbf{N}}} \in {\mathbb{C}^{M \times \tau }}$ is the spatially and temporally white additive Gaussian noise, $\mathbf{Y} \in {\mathbb{C}^{M \times \tau }}$ is the received training signal, and ${\mathbf{s}} \in {\mathbb{C}^{\tau \times 1 }}$ is the shared pilot sequence.
Define the effective channel $\underline{\mathbf{h}}_1 \triangleq \mathbf{W}_1 {\mathbf{h}}_{1}$. Note that $\underline{\mathbf{h}}_1$ has a reduced size, which is $(M-m) \times 1$.
An LS estimate of $\underline{\mathbf{h}}_1$ is:
\begin{equation}
\underline{\widehat{\mathbf{h}}}_1 = \mathbf{W}_1 {\mathbf{Y}}{{\mathbf{s}}^*}{{({{\mathbf{s}}^T}{{\mathbf{s}}^*})}^{ - 1}},
\end{equation}
The key idea is that channel estimate $\underline{\widehat{\mathbf{h}}}_1$ is coarse, yet it can be used as a modified MRC beamformer as it lies in a subspace orthogonal to the interference and is also aligned with the signal subspace of ${\mathbf{h}}_{1}$.
During uplink data transmission phase:
\begin{equation}\label{Eq:Ul2Cell}
{{\mathbf{y}}} =  {{\mathbf{h}}_{1}}{\mathbf{s}}_1^T + \sum_{k=2}^{K}{{\mathbf{h}}_{k}}{\mathbf{s}}_k^T + {{\mathbf{n}}},
\end{equation}
where ${\mathbf{s}}_1, {\mathbf{s}}_2, \cdots, {\mathbf{s}}_K \in \mathbb{C}^{\tau_u \times 1}$ are the transmitted signal sequence. ${\mathbf{y}}, {\mathbf{n}} \in \mathbb{C}^{M \times \tau_u }$ are the received signal and noise respectively.
The subspace-based MRC beamformer is $\underline{\widehat{\mathbf{h}}}^H_1 \mathbf{W}_1$:
\begin{align}\label{Eq:subMRC}
\underline{\widehat{\mathbf{h}}}^H_1 \mathbf{W}_1 {{\mathbf{y}}} =  \underline{\widehat{\mathbf{h}}}^H_1 \underline{\mathbf{h}}_1 {\mathbf{s}}_1^T + & \underbrace{\underline{\widehat{\mathbf{h}}}^H_1 \mathbf{W}_1 {\sum_{k=2}^{K}{{\mathbf{h}}_{k}}{\mathbf{s}}_k^T}} + \underline{\widehat{\mathbf{h}}}^H_1 \mathbf{W}_1 {{\mathbf{n}}}. \\
& \quad \quad \quad \approx \mathbf{0} \nonumber
\end{align}
In case there is no null space for $\mathbf{R}_I$, e.g., the number of users $K$ is large or the interference users have rich scattering environments, the subspace-based method can still avoid the strong eigen modes of interference and therefore reject a good amount of interference.

Note that the subspace projection method has a certain similarity with \cite{muller2013arXiv} which also uses eigen-value decomposition in order to perform blind channel estimation. However there are two main differences: 1) They address only the case of classical massive arrays, not distributed antenna arrays; 2)They use the received power levels domain to separate desired channel and interfering channels. In our approach, the discrimination against interference is related to the phases with which the interference and desired signals arrive at the array. In fact, the two techniques could in principle be combined.


\section{Numerical Results}\label{numericalResult}
We first consider the channel estimation quality in a random network with radius $L=500$ meters. The path loss exponent $\gamma=2.5$. The scattering radius is $r=15$ meters. $P=50$ scatterers are randomly distributed in the scattering ring, which is centered at the user. Define the channel estimation MSE of the $k$-th user as:
 \begin{equation}\label{Eq:err}
 \text{MSE}_{k} \triangleq 10{\log _{10}}\left( {\frac{{ {\left\| {{{\widehat {\mathbf{h}}}_{k}} - {{\mathbf{h}}_{k}}} \right\|^2} }}{{ {\left\| {{{\mathbf{h}}_{k}}} \right\|^2} }}} \right).
\end{equation}
In the simulation we average the channel estimation MSE over different users in order to obtain MSE curve.

In Fig. \ref{fig:MSE.vs.Distance}, we assume the target user is located at the origin while an interfering user (they share the same pilot sequence) is moving over the horizontal axis at increasing distances from user 1. As we can observe, when the MMSE estimator (\ref{Eq:EstimatorDesired}) is used, the channel estimation error is a monotonous decreasing function of the distance between the desired user and the interference user. One may also notice the constant performance gap between LS and MMSE estimator in interference-free scenario, which indicates that covariance information is still helpful even in a highly distributed antenna system. As shown in the blue curve on the top, an LS estimator is unable to separate the desired channel and the interference channel. In contrast, an MMSE estimator has much better performance as its MSE is decreasing almost linearly with inter-user spacing, hence confirming our claims.

\begin{figure}[h]
  \centering
  \includegraphics[width=3.2in]{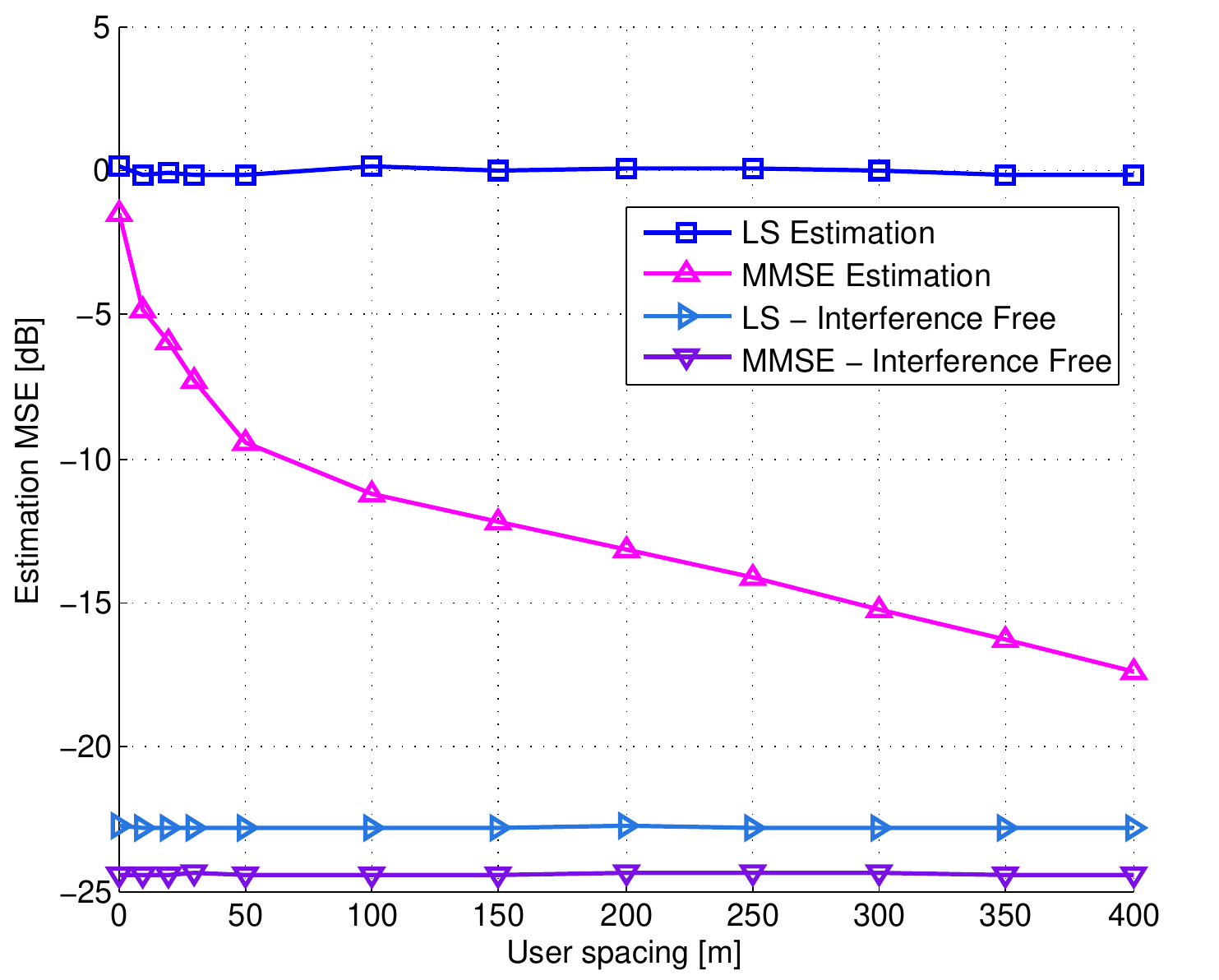}\\
  \caption{Estimation performance vs. distance between two users, $M = 2000$, $r=15$m, single-cell network.} \label{fig:MSE.vs.Distance}
\end{figure}

We then examine the performance of four MRC beamformers in terms of uplink sum-rate in a single-cell setting (Fig. \ref{fig:Sumrate.vs.Distance}) and per-cell rate in a multi-cell setting (Fig. \ref{fig:Rate.vs.Radius}).
The sum-rate is defined as follows:
\begin{equation}\label{Eq:sumrate}
\text{sum-rate} \triangleq { {\sum\limits_{k = 1}^K {{\text{log}}_2 (1  +  \text{SINR}_{k})}} },
\end{equation}
where $K$ is the number of simultaneously served users, and ${\text{SINR}}_{k}$ is the uplink signal-to-noise-plus-interference ratio (SINR) of the $k$-th user.

In Fig. \ref{fig:Sumrate.vs.Distance}, we show the performance of subspace-based MRC beamforming in a single-cell network where two users share the same pilot. The total number of distributed antennas is 500. In the figure ``LS + MRC" denotes the sum-rate performance of MRC beamforming using the LS channel estimate, while the curve ``MMSE + MRC" is the performance of MRC beamforming using the MMSE estimate (\ref{Eq:EstimatorDesired}). ``MMSE + MMSE" denotes the performance curve of MMSE beamforming using MMSE channel estimate when channel covariances (including the interference covariances) are assumed known during both channel estimation and signal detection. The simulation shows the simple subspace-based method has a very good performance.
Due to pilot contamination, the MRC beamformer using MMSE channel estimate is not as good as subspace-based method. The reason is that $\mathbf{R}_1$ and $\mathbf{R}_2$ generally have overlapping signal subspaces here. We may also notice that the subspace-based MRC beamformer has some slight performance gains over the MMSE beamformer.
\begin{figure}[h]
  \centering
  \includegraphics[width=3.2in]{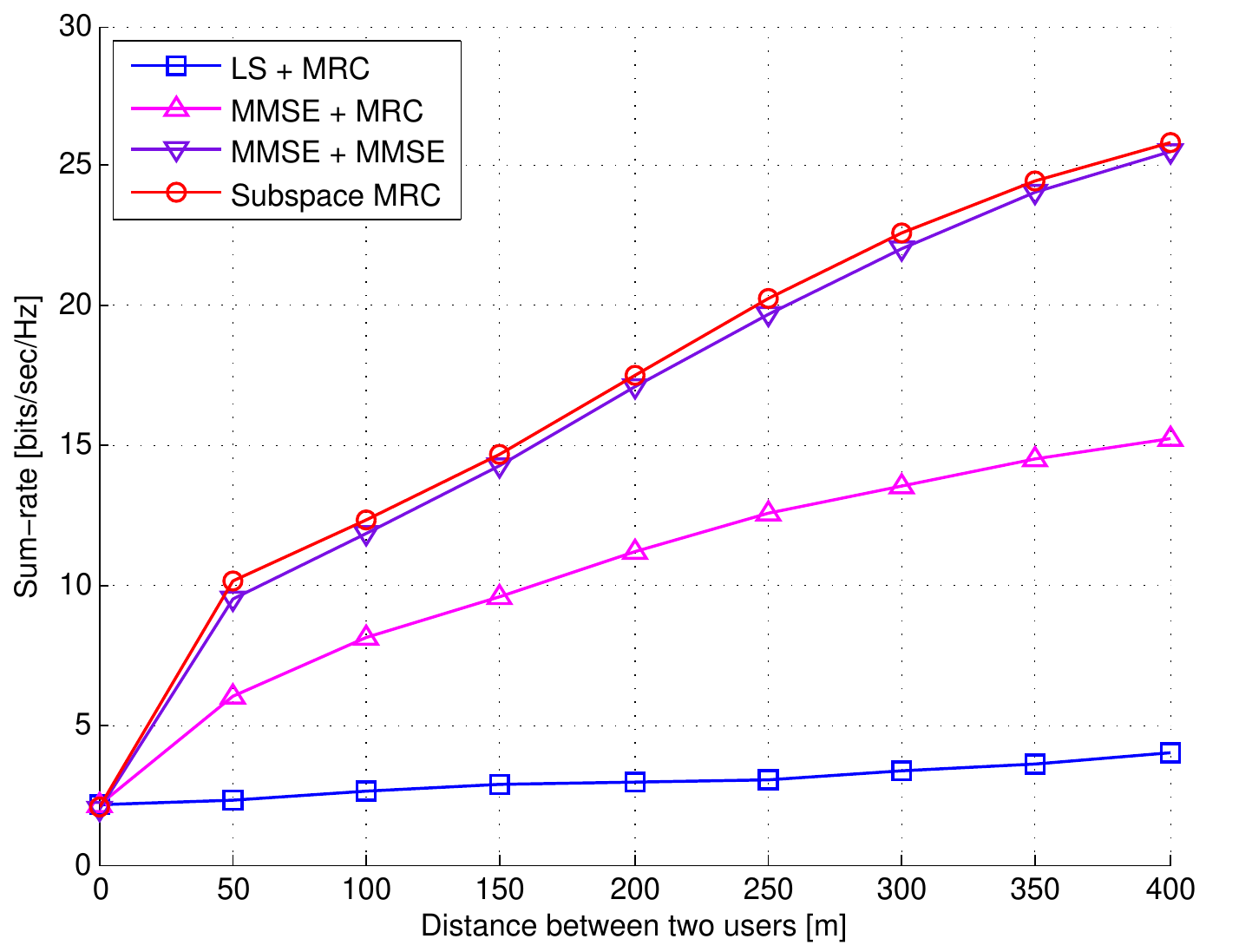}\\
  \caption{Uplink sum-rate vs. distance between 2 users, $M=500$, $r=15$m, cell-edge SNR 20dB, single-cell network.} \label{fig:Sumrate.vs.Distance}
\end{figure}

Fig. \ref{fig:Rate.vs.Radius} depicts the uplink per-cell rate achieved by the above-mentioned MRC beamformers as a function of scattering radius $r$. In the simulation we have 7 hexagonal cells with one center cell and 6 surrounding cells. Each cell has one user. All the users share the same pilot sequence. The per-cell rate is defined as the sum-rate (\ref{Eq:sumrate}) divided by the number of cells.
As can be seen, the subspace-based beamforming shows performance gains over other traditional MRC methods especially when the radius of scattering ring is smaller. It also shows more robustness than MMSE beamformer when the radius of the scattering ring is larger.
\begin{figure}[h]
  \centering
  \includegraphics[width=3.2in]{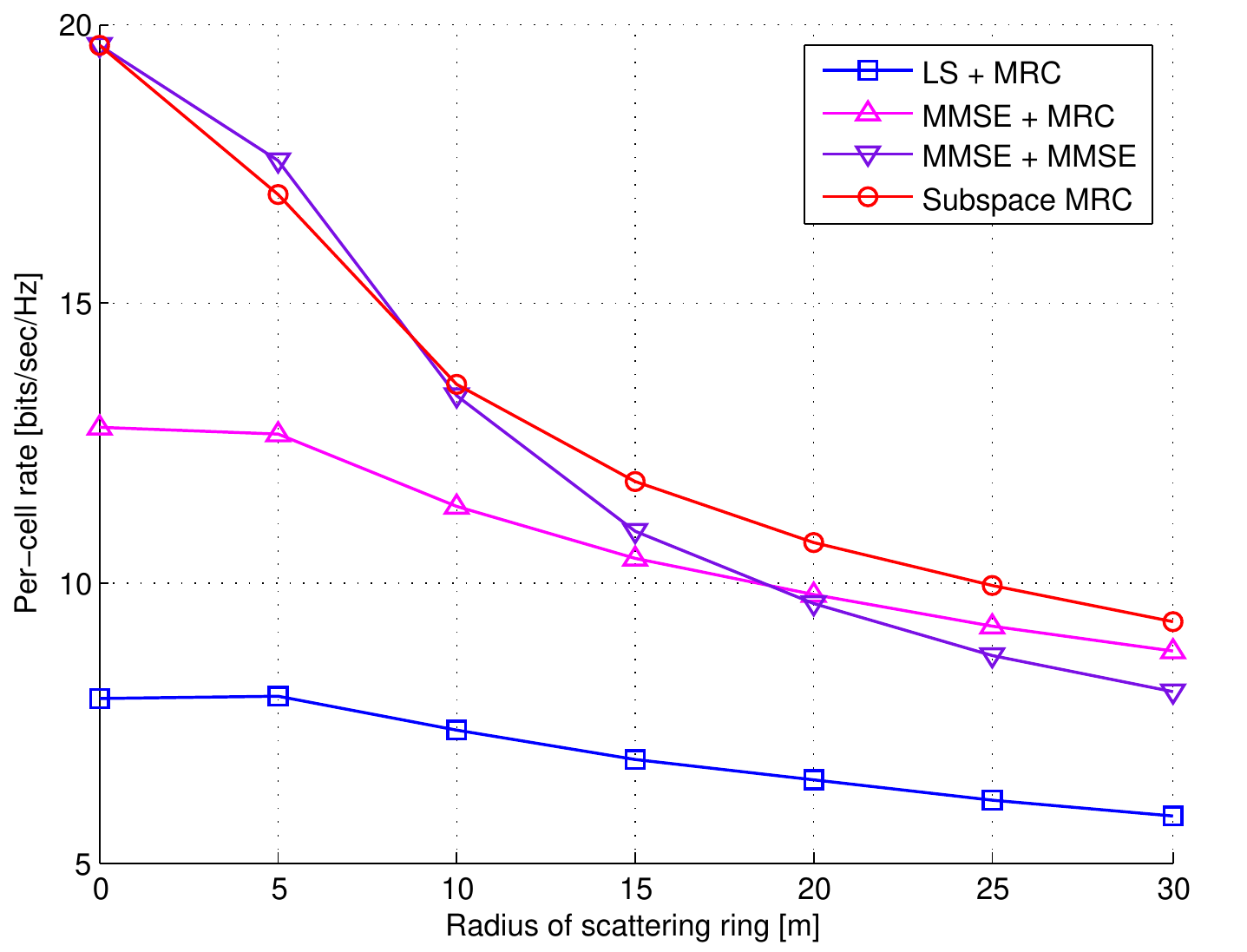}\\
  \caption{Uplink per-cell rate vs. $r$, cell-edge SNR 20dB, 7-cell network, each cell has $M=500$ distributed antennas.} \label{fig:Rate.vs.Radius}
\end{figure}
\section{Conclusions}\label{conclusion}
We investigate low-dimensional properties of covariance signal subspaces in general topologies of massive arrays. We extend previous results known in the uniform linear array case to the case of arrays with random placement, including the case of scattered antennas over a 2D dense network. A correlation model is derived which is exploited to gain insight on the interference rejection capability of low-complexity matched filter-based receivers in distributed antenna settings.

\appendix

\subsection{Proof of Proposition \ref{propDim_nonCali}:}\label{proof:propDim_nonCali}
\begin{proof}
We first consider an $N$-antenna ULA with aperture $\mathcal{D}$ and antenna spacing $D = \mathcal{D}/(N-1)$.
Define
\begin{align*}
\tilde{\boldsymbol{\alpha }}(x) & \triangleq \left[0, {{e^{ - j2\pi \frac{D}{\lambda }x}}}, \cdots,  {{e^{ - j2\pi \frac{D (N-1)}{\lambda }x}}}\right]^T \\
& = \left[0, {{e^{ - j2\pi \frac{\mathcal{D}}{\lambda (N-1)}x}}}, \cdots,  {{e^{ - j2\pi \frac{\mathcal{D} (N-1)}{\lambda (N-1)}x}}}\right]^T.
\end{align*}
Now we define $\tilde{\mathcal{A}} \triangleq \operatorname{span} \{ \tilde{\boldsymbol{\alpha }}{(x)}, {x} \in [{b_1},{b_2}]\}$.
Recall from \cite{yin2012jsac} the following result:

If ${\boldsymbol{\beta }}{\text{(x)}} \triangleq {[\begin{array}{*{20}{c}}
  1&{{e^{ - j\pi x}}}& \cdots &{{e^{ - j\pi (N - 1)x}}}
\end{array}]^T}$. Given ${b_1},{b_2} \in [ - 1,1]$ and ${b_1}<{b_2}$, define $\mathcal{A} \triangleq \operatorname{span} \{ {\boldsymbol{\beta }}{\text{(x)}}, {\text{x}} \in [{b_1},{b_2}]\}$, when $N$ is large,
\begin{equation}\label{Eq:JSAC_LemmaDim_1}
\dim \{ \mathcal{A}\} = \frac{({b_2} - {b_1})N}{2} + o(\frac{({b_2} - {b_1})N}{2}).
\end{equation}

The above conclusion can directly apply: when $N$ is large,
\begin{align*}
\dim \{ \tilde{\mathcal{A}}\} &= \frac{N \mathcal{D}}{(N-1)\lambda}({b_2} - {b_1}) + o\left(\frac{N \mathcal{D}}{(N-1)\lambda}({b_2} - {b_1})\right) \\
&= \frac{{D} N}{\lambda}({b_2} - {b_1}) + o\left(\frac{{D} N}{\lambda}({b_2} - {b_1})\right) \\
&= \frac{M\overline{D}}{\lambda}({b_2} - {b_1}) + o\left(\frac{M\overline{D}}{\lambda}({b_2} - {b_1})\right) \\
&= \frac{M\overline{D}}{\lambda}({b_2} - {b_1}) + o(M).
\end{align*}

We can observe that $\dim \{ \tilde{\mathcal{A}}\}$ has no dependency on $N$. Imagine for any finite aperture $\mathcal{D}$, we let $N \rightarrow \infty$ so that $D \rightarrow 0$.
In this case, all elements of
${\boldsymbol{\alpha }}{(x)}$ can be seen as $M$ (finite) random samples in the vector $\tilde{\boldsymbol{\alpha }}(x) $. Hence
\begin{equation*}
\dim \{ \mathcal{B}\} \leq \dim \{ \tilde{\mathcal{A}}\} = \frac{M\overline{D}}{\lambda}({b_2} - {b_1}) + o(M).
\end{equation*}
Now consider the space $\mathcal{C}$. We define $ \mathcal{C}_q \triangleq  \operatorname{span} \{ {\boldsymbol{\alpha }}{(x)}, {x} \in [{b^{\text{min}}_q},{b^\text{max}_q}]\}$.
An upper bound of its dimension can be obtained by considering the extreme case when all $Q$ spaces are mutually orthogonal so that their dimensions can add up:
\begin{equation*}
\dim \{ \mathcal{C}\} \leq \sum_{q=1}^{Q}\dim\{{\mathcal{C}_q}\}= \sum_{q=1}^{Q}\frac{M\overline{D}}{\lambda}({b^\text{max}_q} - {b^\text{min}_q}) + o(M).
\end{equation*}
Thus, Proposition \ref{propDim_nonCali} is proven.
\end{proof}

\subsection{Proof of Theorem \ref{theoremRankDAS}:}\label{proof:theoremRankDAS}
\begin{proof}

For ease of exposition we omit the user index $k$. Imagine a special case when the scatterers are located in a line which has the length $\tilde{L}$, as shown in Fig. \ref{fig:Proof_Rank_UB}. Assume the antennas are far away so that the scatterers are in the same planar wavefront region. We denote the right end of the scattering line as the reference point. The $m$-th antenna is located $\tilde{d}_m$ meters away from the reference point, at the angle $\theta_m$. The $p$-th scatter is $\tilde{l}(p)$ meters away from the reference point. $\tilde{l}(p)$ follows a uniform distribution, i.e., $\tilde{l}(p) \sim \mathcal{U}(0, \tilde{L})$. The phase shift between the scatterer $p$ and the reference point is $2\pi \tilde{l}(p) \cos(\theta_m)/\lambda$, $1 \leq m \leq M$.
\begin{figure}[h]
  \centering
  \includegraphics[width=3.2in]{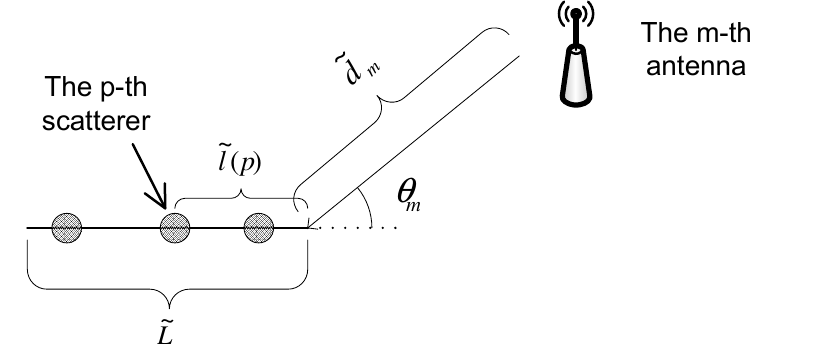}\\
  \caption{Illustration of a line of scatterers.} \label{fig:Proof_Rank_UB}
\end{figure}

Define a diagonal phase matrix
$$\mathbf{\Xi} \triangleq {\mathop{\rm diag}\nolimits} \{ {{e^{-j 2\pi \frac{\tilde{d}_1+r}{\lambda}},...,e^{-j 2\pi \frac{\tilde{d}_M+r}{\lambda}}}}\}.$$
The $p$-th scattering path vector channel is now given by:
\begin{equation*}
{\mathbf{h}_{p}} \triangleq  \left[ {\begin{array}{*{20}{c}}
  {e^{ - j2\pi \frac{d_{p1}+r}{\lambda }}} \\
   \vdots  \\
  {e^{ - j2\pi \frac{d_{pM}+r}{\lambda }}}
\end{array}} \right] e^{j\varphi_{p}} =
e^{j\varphi_{p}} \mathbf{\Xi} \left[ {\begin{array}{*{20}{c}}
  {e^{ - j2\pi \frac{\tilde{l}(p)\cos(\theta_1)}{\lambda }}} \\
   \vdots  \\
  {e^{ - j2\pi \frac{\tilde{l}(p)\cos(\theta_M)}{\lambda }}}
\end{array}} \right]
\end{equation*}
If we define $\tilde{\mathbf{h}}_p \triangleq [{e^{ - j2\pi \frac{\tilde{l}(p)\cos(\theta_1)}{\lambda }}}, \cdots, {e^{ - j2\pi \frac{\tilde{l}(p)\cos(\theta_M)}{\lambda }}}]^T$, we may see that $\tilde{\mathbf{h}}_p$ and $\mathbf{h}_{p}$ are unitarily equivalent, since $e^{j\varphi_{p}} \mathbf{\Xi}$ forms a unitary matrix. If we vary $\tilde{l}(p)$, we can obtain two linear spaces spanned by $\mathbf{h}_{p}$ and/or $\tilde{\mathbf{h}}_p$. The dimensions of the two spaces are equal. Therefore in the following we will find the dimension spanned by $\tilde{\mathbf{h}}_p$ instead of $\mathbf{h}_{p}$.
Assume $N$ is a large integer, we define $x_n \triangleq -1 + \frac{2(n-1)}{N}, n = 1,..., N$, the set $\mathcal{X} \triangleq \{x_n\}$, as well as the vector
\begin{equation*}
\boldsymbol{\mu}_p\triangleq
\left[ {\begin{array}{*{20}{c}}
  {e^{ - j2\pi \frac{\tilde{l}(p)x_1}{\lambda }}} \\
  {e^{ - j2\pi \frac{\tilde{l}(p)x_2}{\lambda }}} \\
   \vdots  \\
  {e^{ - j2\pi \frac{\tilde{l}(p)x_N}{\lambda }}}
\end{array}} \right]
= e^{j2\pi \tilde{l}(p) }\left[ {\begin{array}{*{20}{c}}
  {e^{ - j2\pi \frac{\tilde{l}(p)(0)}{\lambda }}} \\
  {e^{ - j2\pi \frac{\tilde{l}(p)(\frac{2}{N})}{\lambda }}} \\
   \vdots  \\
  {e^{ - j2\pi \frac{\tilde{l}(p)(\frac{2(N-1)}{N})}{\lambda }}}
\end{array}} \right].
\end{equation*}

%

Since $\tilde{l}(p) \sim \mathcal{U}(0, \tilde{L})$, we reuse the result shown in (\ref{Eq:JSAC_LemmaDim_1}):
\begin{align*}
\quad& \dim \{ \operatorname{span} \{ \boldsymbol{\mu}_p, \tilde{l}(p) \in [0, \tilde{L}]\} \} \\
 &= \frac{4\tilde{L}}{N\lambda} \frac{N}{2} + o\left(\frac{4\tilde{L}}{N\lambda} \frac{N}{2}\right) \\
 &= \frac{2\tilde{L}}{\lambda} + o(\tilde{L}).
\end{align*}

Note that the above result holds when $N$ is arbitrarily large. We again observe that when $N \rightarrow \infty$, any $\cos(\theta_m)$ will fall into the set $\mathcal{X}$, which indicates
\begin{equation}\label{Eq:DimSpanHp}
\dim \{ \operatorname{span} \{ \tilde{\mathbf{h}}_p, \tilde{l}(p) \in [0, \tilde{L}]\} \} \leq \frac{2\tilde{L}}{\lambda} + o(\tilde{L}).
\end{equation}
Recall that the covariance matrix $\mathbf{R}= \mathbb{E}\{\frac{1}{{P}}\sum\limits_{p = 1}^P  \sum\limits_{q = 1}^P {{{\mathbf{h}}_{p}}}{{{\mathbf{h}}_{q}^H}}\}$.
Because of the random and independent phases $\varphi _p$ and $\varphi _q$,
$$\forall p \neq q, \mathbb{E}\{{\mathbf{h}}_{p}{{\mathbf{h}}_{q}^H}\} = \mathbf{0}.$$
\begin{equation*}
\mathbf{R} = \mathbb{E}\{\frac{1}{{P}}\sum\limits_{p = 1}^P  {{{\mathbf{h}}_{p}}}{{{\mathbf{h}}_{p}^H}}\} = \mathbb{E}\{ {\mathbf{h}}_{p}{{\mathbf{h}_{p}^H}}\}.
\end{equation*}
We can see that the number of scatterers has no impact on the rank of channel covariance matrix. Hence according to (\ref{Eq:DimSpanHp}), the rank of $\mathbf{R}$ is upper bounded by
$$\text{rank}(\mathbf{R}) \leq \frac{2\tilde{L}}{\lambda} + o(\tilde{L}).$$
Returning to the one-ring model, 
we can interpret the ring as the sum of lines, with the total length $2 \pi r$. 
An extreme case is when all of the channels corresponding to different pieces of the ring span orthogonal spaces, i.e., the rank of the covariance matrix is the sum of the spatial dimensions
corresponding to every pieces of the ring. This is the case when the covariance rank is maximized. Therefore the rank is upper bounded by:
\begin{equation*}
\text{rank}(\mathbf{R}) \leq \frac{4\pi r}{\lambda} + o(r).
\end{equation*}
Thus Theorem \ref{theoremRankDAS} is proven.
\end{proof}

\subsection{Proof of Proposition \ref{propBessel}\label{proof:propBessel}}
\begin{proof}
\quad \emph{Proof:} We split up the disk centered in the midpoint between the two scatterers into $N$ (a large number) equal-sized sectors like dividing up a cake. The BS antennas all fall into one of these sectors. We have
\begin{align*}
\left| {\mathbf{h}}_{2q}^H{{\mathbf{h}}_{1p}} \right| &= \mathop {\lim }\limits_{N \to \infty } \left| \sum\limits_{n = 1}^N {\sum\limits_{m = 1}^{{M_n}} {\left( {h_{2qnm}^*{h_{1pnm}}} \right)} } \right| \\
&= \mathop {\lim }\limits_{N \to \infty } \left| \sum\limits_{n = 1}^N {\sum\limits_{m = 1}^{{M_n}} {\left( {\sqrt{{\beta _{2qnm}}{\beta _{1pnm}}}{e^{j2\pi \frac{{{d_{2qnm}} - {d_{1pnm}}}}{\lambda }}}} \right)} }\right| ,
\end{align*}
where $M_n$ (could be zero) is the total number of antennas located in the $n$-th sector, $h_{2qnm}$ and $h_{1pnm}$ are the channel coefficient between the $m$-th antenna located in the $n$-th sector and the two scatterers. If $N$ is large, the angle contained by the two sides of a sector is small. Therefore BS antennas located in the same sector share the same difference of distances between the two scatterers, i.e., ${d_{2qnm}} -{d_{1pnm}}$ only depends on the sector index $n$.
Based on the fact that $D_{pq}$ is small, we assume ${\beta _{2qnm}} \approx {\beta _{1pnm}}$ and ${d_{2qnm}} -{d_{1pnm}} = D_{pq} \cos(\alpha_n)$, where $\alpha_n \triangleq \frac{2\pi n}{N}$ is the AOA from the BS antennas in the $n$-th sector. 
\begin{align*}
|{\mathbf{h}}_{2q}^H{{\mathbf{h}}_{1p}}| &= \mathop {\lim }\limits_{N \to \infty } \left| \sum\limits_{n = 1}^N {\left( {{e^{j2\pi \frac{{D_{pq} \cos ({\alpha _n})}}{\lambda }}}\sum\limits_{m = 1}^{{M_n}} { \sqrt{{{\beta _{2qnm}}{\beta _{1pnm}}}} } } \right)} \right| \\
&\approx \mathop {\lim }\limits_{N \to \infty } \left| \sum\limits_{n = 1}^N {\left( {{e^{j2\pi \frac{{D_{pq} \cos ({\alpha _n})}}{\lambda }}}\sum\limits_{m = 1}^{{M_n}} { {\beta _{1pnm}} } } \right)} \right|,
\end{align*}
$$\left| {{{\mathbf{h}}_{1p}}} \right|\left| {{{\mathbf{h}}_{2q}}} \right| \approx {\left| {{{\mathbf{h}}_{1p}}} \right|^2} = \mathop {\lim }\limits_{N \to \infty } \sum\limits_{n = 1}^N {\sum\limits_{m = 1}^{{M_n}} {\beta _{1pnm}} }.$$
Due to the symmetry of the network, when $M \rightarrow \infty$, the radio waves can arrive from any direction with equal probability. 
Thus $\sum\nolimits_{m = 1}^{{M_n}} {\beta _{1pnm}}$ is independent of sector index $n$. 
\begin{align*}
\mathop {\lim }\limits_{M \to \infty } \frac{|{\mathbf{h}}_{2q}^H{{\mathbf{h}}_{1p}}|}{{\left| {{{\mathbf{h}}_{1p}}} \right|\left| {{{\mathbf{h}}_{2q}}} \right|}} &\approx \mathop {\lim }\limits_{N \to \infty } \left|\frac{1}{N}\sum\limits_{n = 1}^N {{e^{j2\pi \frac{{{D_{pq}}\cos ({\alpha _n})}}{\lambda }}}} \right|  \\
&= \frac{1}{{2\pi }}\left| \int_0^{2\pi } {{e^{j2\pi \frac{{D_{pq} \cos (\alpha )}}{\lambda }}}} d\alpha \right|\\
&= \left| {J_0}(\frac{{2\pi D_{pq}}}{\lambda })\right|,
\end{align*}
and Proposition \ref{propBessel} is proven.
\end{proof}


\ifCLASSOPTIONcaptionsoff
  \newpage
\fi

%

%

\bibliography{bib/allCitations}

\begin{thebibliography}{10}
\providecommand{\url}[1]{#1}
\csname url@samestyle\endcsname
\providecommand{\newblock}{\relax}
\providecommand{\bibinfo}[2]{#2}
\providecommand{\BIBentrySTDinterwordspacing}{\spaceskip=0pt\relax}
\providecommand{\BIBentryALTinterwordstretchfactor}{4}
\providecommand{\BIBentryALTinterwordspacing}{\spaceskip=\fontdimen2\font plus
\BIBentryALTinterwordstretchfactor\fontdimen3\font minus
  \fontdimen4\font\relax}
\providecommand{\BIBforeignlanguage}[2]{{%
\expandafter\ifx\csname l@#1\endcsname\relax
\typeout{** WARNING: IEEEtran.bst: No hyphenation pattern has been}%
\typeout{** loaded for the language `#1'. Using the pattern for}%
\typeout{** the default language instead.}%
\else
\language=\csname l@#1\endcsname
\fi
#2}}
\providecommand{\BIBdecl}{\relax}
\BIBdecl

\bibitem{marzetta2010}
T.~L. Marzetta, ``Noncooperative cellular wireless with unlimited numbers of
  base station antennas,'' \emph{IEEE Trans. Wireless Commun}, vol.~9, no.~11,
  pp. 3590--3600, Nov. 2010.

\bibitem{rusek2013}
F.~Rusek, D.~Persson, B.~K. Lau, E.~G. Larsson, T.~L. Marzetta, O.~Edfors, and
  F.~Tufvesson, ``Scaling up {MIMO}: Opportunities and challenges with very
  large arrays,'' \emph{IEEE Signal Process. Mag.}, vol.~30, no.~1, pp. 40--60,
  2013.

\bibitem{hoydis2013}
J.~Hoydis, S.~ten Brink, and M.~Debbah, ``Massive {MIMO} in the {UL}/{DL} of
  cellular networks: How many antennas do we need?'' \emph{IEEE J. Sel. Areas
  Commun.}, vol.~31, no.~2, pp. 160--171, 2013.

\bibitem{larsson2013}
E.~G. Larsson, F.~Tufvesson, O.~Edfors, and T.~L. Marzetta, ``Massive {MIMO}
  for next generation wireless systems,'' vol.~52, no.~2, pp. 186--195, Feb.
  2014.

\bibitem{gesbert2010JSAC}
D.~Gesbert, S.~Hanly, H.~Huang, S.~Shamai~Shitz, O.~Simeone, and W.~Yu,
  ``Multi-cell {MIMO} cooperative networks: A new look at interference,''
  \emph{IEEE J. Sel. Areas Commun.}, vol.~28, no.~9, pp. 1380--1408, Dec. 2010.

\bibitem{Lozano2013}
A.~Lozano, R.~Heath, and J.~Andrews, ``Fundamental limits of cooperation,''
  \emph{IEEE Trans. Inf. Theory}, vol.~59, no.~9, pp. 5213--5226, 2013.

\bibitem{jose2011TWC}
J.~Jose, A.~Ashikhmin, T.~L. Marzetta, and S.~Vishwanath, ``Pilot contamination
  and precoding in multi-cell {TDD} systems,'' \emph{IEEE Trans. Wireless
  Commun.}, vol.~10, no.~8, pp. 2640--2651, Aug. 2011.

\bibitem{ngo2011ICASSP}
H.~Q. Ngo, T.~L. Marzetta, and E.~G. Larsson, ``Analysis of the pilot
  contamination effect in very large multicell multiuser {MIMO} systems for
  physical channel models,'' in \emph{Proc. IEEE International Conference on
  Acoustics, Speech and Signal Processing (ICASSP¡¯11)}, Prague, Czech
  Republic, May 2011, pp. 3464--3467.

\bibitem{bjornson2013}
\BIBentryALTinterwordspacing
E.~Bj{\"o}rnson, J.~Hoydis, M.~Kountouris, and M.~Debbah, ``Massive {MIMO}
  systems with non-ideal hardware: energy efficiency, estimation, and capacity
  limits,'' \emph{Submitted to IEEE Trans. Inf. Theory}, 2013. [Online].
  Available: \url{http://arxiv.org/abs/1307.2584}
\BIBentrySTDinterwordspacing

\bibitem{yin2012jsac}
H.~Yin, D.~Gesbert, M.~Filippou, and Y.~Liu, ``A coordinated approach to
  channel estimation in large-scale multiple-antenna systems,'' \emph{IEEE J.
  Sel. Areas Commun.}, vol.~31, no.~2, pp. 264--273, Feb. 2013.

\bibitem{muller2013arXiv}
\BIBentryALTinterwordspacing
R.~R. M{\"u}ller, L.~Cottatellucci, and M.~Vehkaper{\"a}, ``Blind pilot
  decontamination,'' \emph{Submitted to IEEE J. Sel. Topics Signal Process},
  2013. [Online]. Available: \url{http://arxiv.org/abs/1309.6806}
\BIBentrySTDinterwordspacing

\bibitem{ngo2013multicell}
H.~Ngo, E.~Larsson, and T.~Marzetta, ``The multicell multiuser {MIMO} uplink
  with very large antenna arrays and a finite-dimensional channel,'' \emph{IEEE
  Trans. Commun.}, pp. 2350--2361, 2013.

\bibitem{caire2012joint}
A.~Adhikary, J.~Nam, J.-Y. Ahn, and G.~Caire, ``Joint spatial division and
  multiplexing -- the large-scale array regime,'' \emph{IEEE Trans. Inf.
  Theory}, vol.~59, no.~10, pp. 6441--6463, Oct 2013.

\bibitem{jakes1974}
W.~C. Jakes, \emph{Mobile microwave communication}.\hskip 1em plus 0.5em minus
  0.4em\relax Wiley, 1974.

\bibitem{Shiu2000}
D.~shan Shiu, G.~Foschini, M.~Gans, and J.~Kahn, ``Fading correlation and its
  effect on the capacity of multielement antenna systems,'' \emph{IEEE Trans.
  Commun.}, vol.~48, no.~3, pp. 502--513, Mar. 2000.

\bibitem{kennedy2007}
R.~A. Kennedy, P.~Sadeghi, T.~D. Abhayapala, and H.~M. Jones, ``Intrinsic
  limits of dimensionality and richness in random multipath fields,''
  \emph{IEEE Trans Signal Process.}, vol.~55, no.~6, pp. 2542--2556, 2007.

\bibitem{molisch2010}
A.~F. Molisch, \emph{Wireless communications}.\hskip 1em plus 0.5em minus
  0.4em\relax Wiley, 2010.

\bibitem{bjornson2010}
E.~Bjornson and B.~Ottersten, ``A framework for training-based estimation in
  arbitrarily correlated {Rician} {MIMO} channels with {Rician} disturbance,''
  \emph{IEEE Trans. Signal Process.}, vol.~58, no.~3, pp. 1807--1820, Mar.
  2010.

\bibitem{tsai2002}
J.~A. Tsai, R.~M. Buehrer, and B.~D. Woerner, ``The impact of {AOA} energy
  distribution on the spatial fading correlation of linear antenna array,'' in
  \emph{Proc. IEEE Vehicular Technology Conference, (VTC¡¯ 02)}, vol.~2, May
  2002, pp. 933--937.

\bibitem{Kay1993}
S.~M. Kay, \emph{Fundamentals of Statistical Signal Processing: Estimation
  Theory}.\hskip 1em plus 0.5em minus 0.4em\relax Englewood Cliffs, NJ:
  Prentice Hall, 1993.

\bibitem{fill:00}
J.~A. Fill and D.~E. Fishkind, ``The {Moore--Penrose} generalized inverse for
  sums of matrices,'' \emph{SIAM Journal on Matrix Analysis and Applications},
  vol.~21, no.~2, pp. 629--635, 2000.

\bibitem{rakha:04}
M.~A. Rakha, ``On the {Moore--Penrose} generalized inverse matrix,''
  \emph{Applied Mathematics and Computation}, vol. 158, no.~1, pp. 185--200,
  2004.

\bibitem{krasikov2006}
I.~Krasikov, ``Uniform bounds for bessel functions,'' \emph{Journal of Applied
  Analysis}, vol.~12, no.~1, pp. 83--91, 2006.

\bibitem{clarke1968}
R.~H. Clarke, ``A statistical theory of mobile-radio reception,'' \emph{Bell
  System Technical Journal}, vol.~47, no.~6, pp. 957--1000, 1968.

\end{thebibliography}
\bibliographystyle{IEEEtran}

\begin{IEEEbiography}
{Haifan Yin}
received the B.Sc. degree in Electrical and Electronic Engineering and the M.Sc. degree in Electronics and Information Engineering from Huazhong University of Science and Technology, Wuhan, China, in 2009 and 2012 respectively. From 2009 to 2011, he had been with Wuhan National Laboratory for Optoelectronics, China, working on the implementation of TD-LTE systems as an R\&D engineer. In September 2012, he joined the Mobile Communications Department at EURECOM, France, where he is now a Ph.D. student. His current research interests include channel estimation, multi-cell cooperative networks, and large-scale antenna systems.
\end{IEEEbiography}

\vfill

\begin{IEEEbiography}
{David Gesbert}(IEEE Fellow) is Professor and Head of the Mobile Communications Department, EURECOM, France.  He obtained the Ph.D degree from Ecole Nationale Superieure des Telecommunications, France, in 1997. From 1997 to 1999 he has been with the Information Systems Laboratory, Stanford University. In 1999, he was a founding engineer of Iospan Wireless Inc, San Jose, Ca.,a startup company pioneering MIMO-OFDM (now Intel). Between 2001 and 2003 he has been with the Department of Informatics, University of Oslo as an adjunct professor. D. Gesbert has published about 200 papers and several patents all in the area of signal processing, communications, and wireless networks.

D. Gesbert was a co-editor of several special issues on wireless networks and communications theory, for JSAC (2003, 2007, 2009), EURASIP Journal on Applied Signal Processing (2004, 2007), Wireless Communications Magazine (2006). He served on the IEEE Signal Processing for Communications Technical Committee, 2003-2008. He authored or co-authored papers winning the 2004 IEEE Best Tutorial Paper Award (Communications Society) for a 2003 JSAC paper on MIMO systems, 2005 Best Paper (Young Author) Award for Signal Proc. Society journals, and the Best Paper Award for the 2004 ACM MSWiM workshop. He co-authored the book ``Space time wireless communications: From parameter estimation to MIMO systems", Cambridge Press, 2006.
\end{IEEEbiography}

\begin{IEEEbiography}
{Laura Cottatellucci} is currently assistant professor at the Dept. of Mobile Communications in Eurecom. She obtained the PhD from Technical University of Vienna, Austria (2006). Specialized in networking at Guglielmo Reiss Romoli School (1996, Italy), she worked in Telecom Italia (1995-2000) as responsible of industrial projects. From April 2000 to September 2005 she worked as senior research in ftw Austria on CDMA and MIMO systems. From October to December 2005 she was research fellow on ad-hoc networks in INRIA (Sophia Antipolis, France) and guest researcher in Eurecom. In 2006 she was appointed research fellow at the University of South Australia, Australia to work on information theory for networks with uncertain topology. Cottatellucci is co-editor of a special issue on cooperative communications for EURASIP Journal on Wireless Communications and Networking and co-chair of RAWNET/WCN3 2009 in WiOpt.  Her research topics of interest are large system analysis of wireless and complex networks based on random matrix theory and game theory.
\end{IEEEbiography}
\vfill

\end{document}